\documentclass[11pt]{article}

\usepackage{fullpage}
\usepackage{amsmath, amsthm, amssymb}
\usepackage{mathtools}
\usepackage{graphicx}
\usepackage{bbm}
\usepackage{enumerate}
\usepackage{natbib}
\usepackage{subfig}
\usepackage{multirow}
\usepackage{float}
\usepackage{tikz}
\usepackage{longtable}
\usepackage{indentfirst}
\usepackage{setspace}
\RequirePackage[colorlinks,citecolor=blue,urlcolor=blue]{hyperref}

\newtheorem{lemma}{Lemma}

\newtheorem{theorem}{Theorem}

\DeclareMathOperator{\diag}{diag}

\DeclareMathOperator{\IB}{IB}
\DeclareMathOperator{\Tr}{Tr}
\DeclareMathOperator{\Cov}{Cov}
\DeclareMathOperator{\AIC}{AIC}

\DeclareMathOperator*{\argmin}{argmin}

\DeclareMathOperator{\MSE}{MSE}
\DeclareMathOperator{\MISE}{MISE}
\DeclareMathOperator{\AMISE}{AMISE}
\definecolor{c1}{rgb}{0,  0, 1}
\definecolor{c2}{rgb}{1,  0, 0}

\begin{document}
%%%%%%%%%%%%%%%%%%%%%%%%%%%%%%%%%%%%%%%%%%%%%%%%%%%%%%%%%%%%%%%%%%%%%%%%%%%%%%

\title{Robust and Adaptive Functional Logistic Regression}
\author{Ioannis Kalogridis}
\date{%
    Department of Mathematics, KU Leuven, Belgium\\[2ex]%
    \today
}

\maketitle
\begin{abstract}

We introduce and study a family of robust estimators for the functional logistic regression model whose robustness automatically adapts to the data thereby leading to estimators with high efficiency in clean data and a high degree of resistance towards atypical observations. The estimators are based on the concept of density power divergence between densities and may be formed with any combination of lower rank approximations and penalties, as the need arises. For these estimators we prove uniform convergence and high rates of convergence with respect to the commonly used prediction error under fairly general assumptions. The highly competitive practical performance of our proposal is illustrated on a simulation study and a real data example which includes atypical observations.

\end{abstract}

{Keywords:}  Functional logistic regression, robustness, regularization, asymptotics

{MSC 2020}:  62R10, 62G35, 62G20

\section{Introduction}
In recent years, advances in technology and improved storage capabilities have enabled applied scientists to observe and record increasingly complex high-dimensional objects, which may be assumed to be realizations of a random function with values on some nicely behaved function space. Such data are commonly referred to as functional data and their study has enjoyed notable popularity following breakthrough works, such as \citet{Ramsay:1991}, \citet{Ramsay:2005} and \citet{Ferraty:2006}. In many cases, it is of interest to relate a random function $X$ with sample paths in $\mathcal{L}^2(\mathcal{I})$ for some compact $\mathcal{I} \subset \mathbb{R}$ to a binary response $Y$ through the model
\begin{equation}
\label{eq:1}
\mathbb{E}\{Y|X\} = H\left(\alpha_0 + \int_{\mathcal{I}} X(t) \beta_0(t) dt\right),
\end{equation}
with $H: \mathbbm{R} \to (0,1)$ a known bijective "link" function and $(\alpha_0, \beta_0) \in \mathbb{R} \times \mathcal{L}^2(\mathcal{I})$ unknown quantities that need to be estimated from the data. This scalar-on-function logistic regression model originally proposed by \citet{James:2002} remains popular to this day owing to a broad range of possible applications, such as cancer gene detection and foetal risk assessment, see \citet{Rat:2002} and \citet{Wei:2014} respectively.

Despite continuous interest by practitioners, the theoretical study of the scalar-on-function logistic regression has significantly lagged  behind and the scalar-on-function linear regression model has instead garnered most of the attention. For such a model, the interested practitioner would readily find a wealth of theoretically grounded estimation techniques at her disposal. We mention, for example, regression on the leading functional principal components \citep{cardot1999functional, Hall:2007}, penalized spline regression \citep{Cardot:2003}, full-rank smoothing spline regression, as proposed and studied by \citet{Crambes:2009} and \citet{Yuan:2010}, as well as hybrid approaches involving both functional principal components and penalized splines \citep{reiss2007functional}. Unfortunately, none of the theoretical results developped in the aforementioned works carry over to the functional logistic regression model and, to the best of our knowledge, the only  theoretical results for that model were developed by \citet{Cardot:2005} and \citet{M:2005}. The former authors established the consistency of penalized likelihood based spline estimators whereas the latter authors proved the consistency and asymptotic normality of estimates derived from a truncated functional principal component expansion of the coefficient function $\beta_0$. 

An important practical disadvantage of both the proposal of \citet{Cardot:2005} and that of \citet{M:2005} is their lack of resistance towards atypical observations. The reason for this lack of resistance is their reliance on the log-likelihood contributions of each observation $(X_i, Y_i)$ so that any badly fitting pair of observations, e.g., an $X_i$ with a large leverage and a misclassified $Y_i$, would significantly distort the log-likelihood function as well as estimates derived from it. As an attempt to address this lack of resistance, \citet{Den:2016} proposed performing logistic regression on the leading robust principal components obtained with the method of \citet{Hubert:2005}. This method is only partially effective in practice, as it only addresses outlying predictor functions and not misclassified responses, which can still lead to severe distortions. More recently, \citet{Mutis:2022} expanded on the proposal of \citet{Den:2016} by incorporating weights based on the Pearson residuals to the log-likelihood thereby obtaining estimates with greater resistance. However, the theoretical properties of none of these methods are understood so that it is impossible to decide on best practices.

In order to fill the consequential vacuum in the literature with respect to estimators lacking either theoretical justification or resistance to outlying observations, this paper introduces and studies a flexible family of penalized lower-rank smoothers that can fulfil both roles. Resistance is achieved through the use of the density power divergence objective function introduced by \citet{Basu:1998}, which we flexibly combine with a lower rank representation of the coefficient function and a penalty term in order to both obtain easily computable estimates that avoid overfitting. We select the tuning parameter of our objective function in a data-driven way so as to achieve high efficiency in clean data and significant resistance to outlying observations. Under general assumptions, which even permit random tuning parameters, we establish the uniform convergence of our estimator and obtain  a useful error decomposition of the prediction error into the variance and the twin biases resulting from modelling and penalization. These theoretical results are the first among robust estimators for the scalar-on-functional logistic model, but our modern methodology may also be used to strengthen older results, such as those of \citet{Cardot:2005} dealing with penalized maximum likelihood estimation.

The rest of the paper is structured as follows. Section~\ref{sec:2} motivates and describes our proposal in detail. Section~\ref{sec:3} is devoted to the study of the asymptotic properties of the proposed class of estimators; we discuss both consistency and rates of convergence with respect to the uniform and prediction errors respectively. Practical selection of the size of the approximating subspace, the penalty parameter and the tuning constant of the objective function are discussed in Section~\ref{sec:4}. Sections~\ref{sec:5} and \ref{sec:6} illustrate the competitive finite-sample performance of our proposal in a Monte Carlo study and in a real data example involving gait analysis, respectively. Finally, section ~\ref{sec:7} contains a short concluding discussion while all proofs are collected in the appendix.

\section{The proposed family of estimators}
\label{sec:2}

In order to describe the proposed family of estimators in detail, it is first necessary to review the concept of density power divergence, as introduced by \citet{Basu:1998}. Consider two densities with respect to the counting measure, $h$ and $f$. For $\kappa_0 \geq 0$ the density power divergence $d_{\kappa_0}(h,g)$ is defined as
\begin{align}
\label{eq:2}
d_{\kappa_0}(h, f) = \begin{dcases} \sum_{y \in \mathbbm{Z}} \left\{ f^{1+\kappa_0}(y) - \left(1+\frac{1}{\kappa_0} \right)h(y) f^{\kappa_0}(y) + \frac{1}{\kappa_0}h^{1+\kappa_0}(y) \right\} & \kappa_0>0 \\ \sum_{y \in \mathbbm{Z}} h(y) \log\{h(y)/f(y) \} & \kappa_0 = 0.  \end{dcases}
\end{align}
It can be shown that $d_{\kappa_0}(h, f) \geq 0$ and that $d_{\kappa_0}(h, f)=0$ if, and only if, $h = f$ on all of $\mathbbm{Z}$, see Lemma~\ref{lem:1} in the Appendix. If we assume that $h$ is the true density and $\{f_{\theta},\theta \in \Theta\}$ is a family of model densities, density power divergence estimators asymptotically minimize $d_{\kappa_0}(h,f_{\theta})$ over $\theta \in \Theta$. It is worth noting that $d_{0}(h,  f_{\theta})$ is  merely the Kullback-Leibler divergence which is asymptotically minimized by maximum likelihood estimators. It is well-known, however, that although maximum likelihood estimators are efficient, they are not robust against outliers and model misspecification. Density power divergence estimators based on $\kappa_0>0$ can be shown to be more resistant, in the sense that their influence function \citep{Hampel:2011} is bounded and redescending for large outliers, and are therefore better suited for the analysis of complex datasets containing outliers. In fact, as demonstrated by the analysis of \citet{Kal:2022}, the robustness of density power divergence estimators increases with higher values of $\kappa_0$, but that comes at a cost for their efficiency so that in practice a deal between efficiency and robustness needs to be struck.

To see how density divergence may be used in the context of functional logistic regression, let us assume that  we are in possession of a sample $(X_1, Y_1), \ldots, (X_n, Y_n)$ consisting of independent and identically distributed tuples in $\mathcal{L}^2(\mathcal{I}) \times \{0,1\}$ following model \eqref{eq:1}. For simplicity and without loss of generality we shall identify $\mathcal{I}$ with the $[0,1]$-interval. In view of \eqref{eq:1}, we may assume that conditionally on $X_i$ each binary response $Y_i$ possesses a Bernoulli density $f_{p_i(\alpha_0,\beta_0)}(y) = p_{i}(\alpha_0, \beta_0)^{y} (1-p_{i}(\alpha_0, \beta_0))^{1-y} $ where $p_i(\alpha_0, \beta_0)$, the probability of "success", is given by
\begin{align*}
p_i(\alpha_0, \beta_0) = H\left(\alpha_0 + \int_{[0,1]} X_i(t) \beta_0(t) dt\right), \quad (i = 1, \ldots, n).
\end{align*}
We show in the appendix that $M((\alpha,\beta), \kappa_0) := \mathbb{E}\{d_{\kappa_0}(f_{ p_1(\alpha_0, \beta_0)}, f_{p_1(\alpha, \beta)}) \}$ is bounded below by $0$ and bounded above by $1+\kappa_0^{-1}$. The upper bound is attained if, and only if, $|p_1(\alpha_0,\beta_0)) - p_1(\alpha,\beta)| = 1$ almost everywhere. Moreover, $M((\alpha,\beta), \kappa_0)$ is a continuous function of $(\alpha, \beta)$ whenever $H$, the link function, is also continuous. These properties make $M((\alpha, \beta), \kappa_0)$ an attractive objective function for robust estimation, but unfortunately minimizing $M((\alpha, \beta), \kappa_0)$ directly is not a feasible task owing to our lack of knowledge of either the distribution of $X$ or the population parameters $(\alpha_0, \beta_0)$. Our approach thus consists of minimizing a suitable approximation of this criterion.

A finite sample approximation to $M((\alpha, \beta), \kappa_0)$ may be obtained by replacing $\mathbb{E}\{\cdot\}$ with its empirical counterpart weighing each $X_i$ with $n^{-1}$, as is typically done in the construction of M-estimators \citep[see, e.g.,][Chapter 5]{VDV:1998}. Next, from \eqref{eq:2} it may be seen that we may drop terms not depending on $(\alpha, \beta)$, as these terms do not influence the value of the minimizer. Lastly, for every $\kappa_0>0$, a conditionally unbiased estimator for $\sum_{y} f^{\kappa_0}_{p_i(\alpha, \beta)}(y) f_{p_i(\alpha_0, \beta_0)}(y)$ is given by $f_{p_i(\alpha, \beta)}^{\kappa_0}(Y_i)$. Putting everything together, we arrive at the finite-sample approximation of $M((\alpha, \beta), \kappa_0)$ given by
\begin{align}
\label{eq:3}
\frac{1}{n}\sum_{i=1}^n l_{\widehat{\kappa}_n}(Y_i, p_i(\alpha, \beta)) =  \frac{1}{n}\sum_{i=1}^n \left[ \sum_{y\in\{0,1\}} f^{1+\widehat{\kappa}_n}_{p_i(\alpha, \beta)}(y)  - \left(1+\frac{1}{\widehat{\kappa}_n} \right) f^{\widehat{\kappa}_n}_{p_i(\alpha, \beta)}(Y_i) \right],
\end{align}
where we now also allow for changing, even random, tuning parameters $\widehat{\kappa}_n$ that may depend on the sample itself. These parameters determine the trade-off between robustness and efficiency, as for $\widehat{\kappa}_n \to 0$ we obtain what is essentially the log-likelihood of the Bernoulli distribution whereas for larger $\widehat{\kappa}_n$ we obtain an objective function that leads to more robust but less efficient estimators. The above formulation allows us in a sense to get the best of both worlds, as $\widehat{\kappa}_n$ may be chosen small enough in clean samples and larger in datasets containing outlying observations.  A possible selection scheme for $\widehat{\kappa}_n$ that we have found to be effective in practice is outlined in Section~\ref{sec:4} below.

We follow the recipe of \citet[Chapter 15]{Ramsay:2005} in order to obtain estimators from the objective function \eqref{eq:3}. That is, rather than minimizing \eqref{eq:3} directly over the whole of $ \mathbbm{R} \times \mathcal{L}^2([0,1])$, which entails the risk of overfitting, we instead restrict attention to the considerably smaller space $\mathbbm{R} \times \Theta_{K_n}$ where $\Theta_{K_n}$ is a $K_n$-dimensional linear subspace that is spanned by known real-valued functions $\theta_1, \ldots, \theta_{K_n}$. These functions can, for example, be B-splines or Fourier bases, both of which are broadly used in functional data analysis. Here, the dimension of the approximating subspace, $K_n$, depends on $n$, as for the approximation bias to be vanish in the limit we need to have $K_n \to \infty$ as $n \to \infty$ but not too fast, see Theorem~\ref{thm:1} below. Other than a finite-dimensional subspace, we also make use of a penalty functional $\mathcal{J}: \Theta_{K_n} \to \mathbbm{R}_{+}$ with the aim of further controlling the roughness of the slope estimates. With these two refinements, our proposal amounts to estimating the unknown quantities with
\begin{align}
\label{eq:4}
(\widehat{\alpha}_n, \widehat{\beta}_n) = \argmin_{(\alpha, \beta) \in \mathbbm{R} \times \Theta_{K_n} } \left[  \frac{1}{n} \sum_{i=1}^n  l_{\widehat{\kappa}_n}(Y_i, p_i(\alpha, \beta))  + \widehat{\lambda}_n \mathcal{J}(\beta) \right],
\end{align}
with $\widehat{\lambda}_n$ a possibly random penalty parameter determining the premium that is placed on the roughness of $\widehat{\beta}_n$, as measured by $\mathcal{J}(\widehat{\beta}_n)$.
For $\widehat{\lambda}_n = 0$ the estimator is an unpenalized linear combination of the $\theta_j$ whereas for $\widehat{\lambda}_n \to \infty$ the penalty dominates the objective function forcing the slope estimator to lie in its null-space. The null-space of $\mathcal{J}(\cdot)$ is often low dimensional, e.g., it may consist of polynomials up to a certain order, so that functions that lie on it tend to be of simple form.

It is worth noting that the estimator in \eqref{eq:4} is regularized both with respect to its dimension and its complexity, as, in our experience, both types of regularization are needed for a reliable all-around practical performance. In particular, the present dual regularization scheme ensures that the estimator can be made flexible enough by taking a large value of $K_n$, while also remaining interpretable on account of the penalization by means of $\mathcal{J}(\cdot)$. Moreover, the present set-up is very flexible and allows for a wide variety of approximating subspaces $\Theta_{K_n}$ and penalties $\mathcal{J}(\cdot)$, as the applied scientist sees fit. These may include the frequently used B-splines with derivative or difference penalties, but also wavelets with total variation or $l_1$ penalties for cases in which $\beta_0$ is suspected to be less smooth. The interested reader is referred to \citet{Kal:2023} for a discussion of a number of practically useful alternatives.

At a first glance it may seem that the proposed estimator is computationally intensive, as it depends on three parameters, namely the tuning parameter $\widehat{\kappa}_n$, the dimension of the approximating subspace, $K_n$, and the penalty parameter $\widehat{\lambda}_n$. However, practical experience with lower-rank penalized estimators has shown that $K_n$ is not as crucial as the other parameters of the problem and therefore it can be selected in a semi-automatic manner, that is, almost independently of the sample \citep[see, e.g.,][]{Ruppert:2003, Wood:2017}. An efficient method of selecting $\widehat{\kappa}_n$ and $\widehat{\lambda}_n$ is discussed in Section~\ref{sec:4} below.

\section{Theoretical properties}
\label{sec:3}

\subsection{Consistency}

We now study the asymptotic properties of penalized density power divergence estimators. Throughout this section we assume that $\alpha_0 = 0$ and the object of interest is the coefficient function $\beta_0$, as is typically the case. The assumptions that we require for our theoretical development are as follows.

\begin{itemize}
\item[(A1)] There exists a $\kappa_0 \in (0, 1]$ such that $\widehat{\kappa}_n \xrightarrow{\mathbb{P}} \kappa	_0$, as $n \to \infty$.
\item[(A2)] The link function $H:  \mathbbm{R} \to (0,1)$ is bijective, strictly increasing and differentiable with its derivative $H^{\prime}$ satisfying $\sup_{x \in \mathbbm{R}} H^{\prime}(x) \leq c_0$ for some $c_0<\infty$.
\item[(A3)] The coefficient function $\beta_0$ belongs to a Banach space of functions, $\mathcal{B}([0,1])$, with norm $\| \cdot \|_{\mathcal{B}}$ that is embeddable in $\mathcal{C}([0,1])$ with its norm $\| \cdot \|_{\infty}$, i.e., $\mathcal{B}([0,1]) \subset \mathcal{C}([0,1])$ and there exists a finite $c_1>0$ such that
\begin{align*}
\|f\|_{\infty} \leq c_1 \|f \|_{\mathcal{B}}, \quad \forall f \in \mathcal{B}([0,1]).
\end{align*} 
\iffalse Furthermore, the unit ball $\{f \in \mathcal{B}([0,1]): \|f\|_{\mathcal{B}} \leq 1 \}$ is compact in the topology of the norm $\|\cdot\|_{\infty}$ \fi
\item[(A4)] $\Theta_{K_n} \subset \mathcal{B}([0,1])$ and $K_n \asymp n^{\gamma}$ for some $\gamma \in (0,1)$. Letting $\widetilde{\beta}_{K_n} = \argmin_{\beta \in \Theta_{K_n}} \|\beta-\beta_0\|$ we require $\|\widetilde{\beta}_{K_n}-\beta_0\| \to 0$ and $\widehat{\lambda}_n \mathcal{J}(\widetilde{\beta}_{K_n}) \xrightarrow{\mathbb{P}} 0$, as $n \to \infty$.
\item[(A5)] There exists $0<C<\infty$ such that $\mathbb{P}(\|X\|\leq C)= 1$ and there exists  $0 \leq c_2<1$ such that for every $\beta \in \mathcal{B}([0,1]$ with $\beta \neq 0$ it holds that $\mathbb{P}( \langle X, \beta \rangle = 0) \leq c_2$.
\end{itemize}

Assumption (A1) requires that the random tuning parameter $\widehat{\kappa}_n$ converges in probability to some positive constant $\kappa_0$. Assumptions of this type are commonly used in the theoretical treatment of estimators involving nuisance (estimated) parameters, see, e.g., \citep[Sections 5.4 and 5.8.1]{VDV:1998}. It is also possible to select the tuning parameter deterministically, i.e., $\widehat{\kappa}_n = \kappa_0$ for some $\kappa_0>0$ for all $n \in \mathbbm{N}$ in which case (A1) is clearly satisfied and all subsequent arguments remain valid.  Assumption (A2) concerning the link function, $H$, is similarly very general and permits all commonly used link functions, namely the logit, probit and Gumbel (complementary log-log) links. 

Assumption (A3) is a mild assumption on the coefficient function, $\beta_0$. It is required that that $\beta_0$ lies in a complete normed space of functions on $[0,1]$, $\mathcal{B}([0,1])$, which is a subspace of the space of continuous functions, $\mathcal{C}([0,1])$. This assumption is considerably more general than the corresponding assumption of \citet{Kal:2023}, as we do not require the embedding of $\mathcal{B}([0,1])$ in $\mathcal{C}([0,1])$ to be compact. As a result, (A3) is satisfied by a wider variety of function spaces. For example, while in our framework we may take $\mathcal{B}([0,1]) = \mathcal{C}([0,1])$, this is not permitted by the assumptions of \citet{Kal:2023}, as the unit ball $\{f \in \mathcal{C}([0,1]): \|f\|_{\infty} \leq 1 \}$ is not compact in $\mathcal{C}([0,1])$.

Assumption (A4) requires that the dimension of the approximating subspace, $K_n$, does not grow too fast relative to the sample size $n$, so that $K_n/n \to 0$. This is a common assumption for estimators based on approximating subspaces (sieves), see, e.g., \citet[Chapter 15]{Eg:2009} and \citet{Kal:2023}. Moreover, it is required that $\beta_0$ can be well approximated in the $\mathcal{L}^2([0,1])$-norm by an element of $\Theta_{K_n}$. This sequence of deterministic approximations of $\beta_0$ in $\Theta_{K_n}$, $\widetilde{\beta}_{K_n}$, should have finite roughness, as measured by $\mathcal{J}(\cdot)$ so that $\widehat{\lambda}_n \mathcal{J}(\widetilde{\beta}_{K_n}) \xrightarrow{\mathbb{P}} 0$. It is thus required that $\mathcal{J}(\cdot)$ is selected appropriately relative to $\Theta_{K_n}$. Both requirements are met by many popular approximating subspaces and penalties, e.g., by B-splines with derivative-based penalties and by Fourier bases with either derivative or harmonic acceleration penalties \citep{Ramsay:2005}. It is worth noting that contrary to other theoretical treatments of functional linear regression models, such as \citep{Cardot:2003,Yuan:2010, Shin:2016}, $\widehat{\lambda}_n$ is herein more realistically treated as a sequence of random variables rather than a sequence of constants. This is an important difference, as most often $\widehat{\lambda}_n$ is selected in a data-driven manner and is thus random rather than fixed.

Finally, assumption (A5) requires that the covariate, $X$, is almost surely bounded when viewed as a random element of $\mathcal{L}^2([0,1])$. A similar assumption exists in many treatments of functional linear regression models, e.g., in \citep{Cardot:2003, Kal:2023}. It is additionally required that $X$ is not concentrated on any finite-dimensional subspace of $\mathcal{B}([0,1])$. Since $\mathcal{B}([0,1]) \subset \mathcal{L}^2([0,1])$ this requirement is met whenever the null space of the covariance operator of $X$ consists only of the zero element $\{0\}$ which holds if, and only if, all of its countably infinite eigenvalues are strictly positive. This assumption is likewise more general than assumption (A4) of \citet{Kal:2023}, as it involves only the space $\mathcal{B}([0,1])$ instead of the larger $\mathcal{L}^2([0,1])$. This is possible because, by (A4), our estimator lies in $\mathcal{B}([0,1])$, hence we can restrict attention to this smaller space of functions.

As shown in the appendix (see Lemma ~\ref{lem:1} and Lemma~\ref{lem:2} there), our estimator is Fisher consistent, that is, it estimates the correct quantity at the population level,  under assumptions (A1), (A2) and (A5) alone. The practical significance of assumptions (A3) and (A4) is that they allow us to additionally obtain  a strong form of consistency of our finite sample estimator $\widehat{\beta}_n$ towards the population function $\beta_0$. This result is presented in Theorem~\ref{thm:1} below.

\begin{theorem}
\label{thm:1}
Suppose that assumptions (A1)--(A5) are satisfied. Then, $\|\widehat{\beta}_n - \beta_0\|_{\mathcal{B}} \xrightarrow{\mathbb{P}} 0$, as $n \to \infty$. 
\end{theorem}

Theorem~\ref{thm:1} establishes the convergence of our estimator $\widehat{\beta}_n$ as an element of the Banach space $\mathcal{B}([0,1])$. To the best of our knowledge, this is a strong mode of convergence that has not been established previously for any robust functional regression estimators. In particular, the results of \citet{Boente:2020} and \citet{Kal:2023} only imply that $\|\widehat{\beta}_n - \beta_0\|_{\mathcal{B}} = O_\mathbb{P}(1)$,  as $n \to \infty$, which is clearly a weaker result. Since, by our assumptions, $\mathcal{B}([0,1]$ is embeddable in $\mathcal{C}([0,1])$, the result of Theorem~\ref{thm:1} minimally implies the uniform convergence $\|\widehat{\beta}_n-\beta_0\|_{\infty} \xrightarrow{\mathbb{P}} 0$, as $n \to \infty$, but it may also imply convergence of certain derivatives if $\mathcal{B}([0,1])$ is more regular, e.g., if $\mathcal{B}([0,1])$ is a Hölder or a Sobolev space, as is often assumed with B-spline or Fourier estimators. For example, if $\mathcal{B}([0,1])$ is the Hölder space with derivatives up to order $m$ satisfying a Hölder condition with exponent $0< r \leq 1$, then convergence is with respect to the norm
\begin{align*}
\left\|f \right\|_{\mathcal{B}} = \left\|f \right\|_{\infty} + \max_{0 \leq  j \leq m} \sup_{\substack{x,y \in [0,1] \\ x \neq y}} \left| \frac{f^{(j)}(x) - f^{(j)}(y)}{(x-y)^{r}} \right|, \quad f \in \mathcal{B}([0,1]).
\end{align*}
Besides its intrinsic theoretical interest, such a mode of convergence offers important practical guarantees regarding the performance of the estimator, in that the population function $\beta_0$ and its derivatives will be estimated equally well throughout their domain.

\subsection{Rates of convergence} 

To obtain a rate of convergence of the proposed family of estimators, we use the distance criterion given by
\begin{align*}
\left|\pi(\widehat{\beta}_n,\beta_0)\right|^2 =  \mathbb{E} \left\{\left| \langle X, \widehat{\beta}_n-\beta_0 \rangle \right|^2 \right\},
\end{align*}
which may be rewritten as $\pi(\widehat{\beta}_n, \beta_0) = \{\langle \Gamma (\widehat{\beta}_n-\beta_0), \widehat{\beta}_n-\beta_0 \rangle\}^{1/2}$ with $\Gamma$ denoting the self-adjoint Hilbert-Schmidt covariance operator of $X$. This is a natural criterion to consider and it has been extensively in theoretical investigations of functional linear models, e.g., by \citet{Cardot:2005, Crambes:2009, Boente:2020} and \citet{Kal:2023}. Intuitively, $\pi(\widehat{\beta}_n, \beta_0)$ corresponds to the average prediction error that arises when using $\langle X_{n+1}, \widehat{\beta}_n \rangle$ to predict $\langle X_{n+1}, \beta_0 \rangle$, where $X_{n+1}$ is a new random function possessing the same distribution as $X$. 

It is worth noting that, under (A5), $\pi(\widehat{\beta}_n, \beta_0)  \leq C \| \widehat{\beta}_n-\beta_0 \|_{\infty}$, so that, by Theorem~\ref{thm:1}, $\pi(\widehat{\beta}_n, \beta_0)  \xrightarrow{\mathbb{P}} 0$.  Theorem~\ref{thm:2} below serves to delineate the roles of the subspace $\Theta_{K_n}$, the penalty functional $\mathcal{J}(\cdot)$ and the penalty parameter $\widehat{\lambda}_n$ in the convergence of the prediction error to $0$.

\begin{theorem}
\label{thm:2}
Suppose that assumptions (A1)--(A5) are satisfied. Then,
\begin{align*}
\left| \pi(\widehat{\beta}_n, \beta_0)\right|^2 = O_{\mathbb{P}}\left(\frac{K_n \log n}{n} \right) +  O_{\mathbb{P}}\left(\left\|\widetilde{\beta}_{K_n}-\beta_0 \right\|^2 \right) + O_{\mathbb{P}}\left(\widehat{\lambda}_n \mathcal{J}\left(\widetilde{\beta}_{K_n}\right) \right), \quad \text{as} \ n \to \infty.
\end{align*}
\end{theorem}
\noindent
Theorem~\ref{thm:2} presents the prediction error as a function of three quantities, which respectively represent the variance, the squared approximation bias and the regularization bias. This decomposition is the same as the decomposition obtained by \citet{Kal:2023} for functional linear regression estimators. However, as discussed earlier, our assumptions with respect to $\beta_0$ and the finite-dimensional distributions of $X$ are significantly weaker than the assumptions of those authors.

Examining the prediction error decomposition in detail, we see that the variance term $K_n \log n/n$ depends only on the size of the approximating subspace, $\Theta_{K_n}$,  but not its type. This situation has a well-known parallel in non-parametric regression \citep[Chapter 15]{Eg:2009}. The $\log n$ term is not usually encountered in the non-parametric setting and appears herein due to the more difficult nature of the problem. In particular, it appears because of the infinite-dimensional predictor $X$. Contrary to the variance, the approximation bias $\|\widetilde{\beta}_{K_n}-\beta_0 \|$ does depend on the type of approximating subspace as well as $\beta_0$ itself. To achieve a good rate of decay of the approximation error, we need to select a subspace that approximates well the class of functions to which $\beta_0$ belongs. For instance, a subspace of twice continuously differentiable functions, if we suspect that $\beta_0$ is also twice continuously differentiable. Lastly, the decay of $\widehat{\lambda}_n \mathcal{J}(\widetilde{\beta}_{K_n})$ to $0$ essentially depends on the appropriateness of $\mathcal{J}(\cdot)$ as a measure of roughness on $\Theta_{K_n}$. 

As a concrete example, let us consider the case wherein $\mathcal{B}([0,1])$ is the space of functions having uniformly bounded derivatives up to order $r \geq 1$ with $r$th derivative satisfying a Hölder condition of order $v \in (0,1)$. For $\Theta_{K_n}$ we select the $(K_n+p)$-dimensional B-spline subspace of order $p>r$ generated by $K_n$ equidistant interior knots in $[0,1]$ and $\mathcal{J}(\beta) = \|\beta^{(q)}\|^2$ for $q<p$. Then, $\mathcal{J}(\widetilde{\beta}_{K_n}) = O(1)$ and, as shown in \citet[p. 149]{DB:2001}, $\|\widetilde{\beta}_{K_n}-\beta_0 \| = O(K_n^{-r-v})$. Hence,
\begin{align*}
\left| \pi(\widehat{\beta}_n, \beta_0)\right|^2 = O_{\mathbb{P}}\left(\frac{K_n \log n}{n} \right) +  O_{\mathbb{P}}\left(K_n^{-2(r+v)} \right) + O_{\mathbb{P}}\left(\widehat{\lambda}_n \right).
\end{align*}
For $K_n \asymp n^{1/(2(r+v)+1)}$ and $\widehat{\lambda}_n = O_{\mathbb{P}}(n^{-c})$ with $c \geq 2(r+v)/(2(r+v)+1)$ we obtain the optimal rate $| \pi(\widehat{\beta}_n, \beta_0)|^2 = O_{\mathbb{P}}( n^{-2(r+v)/(2(r+v)+1)} \log n)$. Interestingly, this rate of convergence is substantially faster than the rate $O_{\mathbb{P}}( n^{-2(r+v)/(4(r+v)+1)})$ obtained by \citet{Cardot:2005} for their penalized spline likelihood estimator. The difference is attributable to our use of modern empirical process theory in the proofs of our results.

\iffalse
\begin{itemize}
\item[(A6)] There exists $0<c_3 < \infty$ such that 
\begin{align*}
\mathbb{P}\left(\frac{1}{c_3} \leq  H\left(\langle X_1, g_0 \rangle \right) \leq 1- \frac{1}{c_3} \right) = 1
\end{align*}
\item[(A7)] There exist $0<c_4, c_5<\infty$ such that
\begin{align*}
\mathbb{P}\left(H^{\prime}\left(x \right) \geq c_4, \forall x: \left|x-\langle X_1,  g_0 \rangle\right| <c_5 \right) = 1.
\end{align*}
\end{itemize}
\noindent
Assumptions (A6) and (A7) are identifiability conditions that are functional data equivalents of the conditions invariably used in the asymptotics of non-robust nonparametric logistic regression, see, e.g., assumptions (A6) and (A7) in \citet{Mammen:1997}. In the case of (A6) we require that the population probability of "success" is almost surely bounded away from $0$ and $1$ and in the case of (A7) that the link function is almost surely strictly increasing in a small neighbourhood about $\langle X_1, g_0 \rangle$. Both assumptions are reasonable and are satisfied quite generally. For instance, under (A5), $|\langle X_1, g_0 \rangle| \leq C \|g_0\|$ with probability one, so that, in combination with (A5), (A7) essentially requires that $H$ is strictly increasing in some compact set about zero, which is a very mild assumption.
\fi

\section{Practical implementation}
\label{sec:4}

\subsection{Computational algorithm}

The penalized density power divergence estimator defined in \eqref{eq:4} depends on the choice of the approximating subspace, $\Theta$, its dimension, $K_n$, the tuning parameter $\widehat{\kappa}_n$ and the penalty parameter $\widehat{\lambda}_n$. In this section we outline an effective strategy for their selection, but first we discuss the computation of penalized density power divergence estimates. To develop our algorithm, we  assume with little loss of generality that the penalty $\mathcal{J}(\beta)$ can be written in quadratic form, i.e., we can write $\mathcal{J}(\beta) = \boldsymbol{\beta}^{\top} \mathbf{P} \boldsymbol{\beta}$ for some positive semi-definite $\mathbf{P}$, with $\boldsymbol{\beta} \in \mathbbm{R}^{K_n}$ the coefficient vector of $\beta \in \Theta_{K_n}$. While many penalties, such as penalties on integrated squared derivatives, fulfil this assumption, $l_1$-type penalties do not. However, for such penalties we may employ the computationally convenient local quadratic approximation proposed by \citet{Fan:2001} so that the core idea of our algorithm remains valid.

With the quadratic simplification of the penalty, we solve \eqref{eq:4} with a modified Fisher-scoring procedure the updating steps of which may be reduced to penalized iteratively reweighted least squares updates, in the manner outlined by \citet{Green:1994}. In particular, the weights in our case are given by
\begin{align*}
w_i = (1+\widehat{\kappa}_n)\frac{\mathbb{E}_{p_i}\{f_{p_i}^{\widehat{\kappa}_n}(Y_i)|Y_i - p_i|^2 \vert X_i\}}{p_i^2(1-p_i)^2} \left| \frac{\partial p_i}{ \partial \eta_i} \right|^2, \quad (i=1, \ldots, n),
\end{align*}
where $\mathbb{E}_{p_i}\{\cdot\}$ denotes expectation with respect to a Bernoulli density with probability of "success" $p_i = H(\eta_i)$ and the linear components $\eta_i = \alpha +  \langle X_i, \beta \rangle, i = 1,\ldots, n$. At the same time, the vector of "working" data $\mathbf{z} \in \mathbbm{R}^n$ is given by
\begin{align*}
z_i = \eta_i - \left(1+ \widehat{\kappa}_n\right)\frac{\partial p_i}{ \partial \eta_i} \frac{\mathbb{E}_{p_i}\{f_{p_i}^{\widehat{\kappa}_n}(Y_i)(Y_i - p_i)/(p_i(1-p_i)) \vert X_i  \} - f_{p_i}^{\widehat{\kappa}_n}(Y_i)(Y_i - p_i)/(p_i(1-p_i))}{w_i}.
\end{align*}
The expressions simplify somewhat when the "canonical" link $H(x) = e^{x}/(1+e^{x})$ is used, as in that case $\partial p_i/\partial \eta_i = p_i(1-p_i)$. We initiate our Fisher scoring algorithm at the density power divergence estimates obtained with $\widehat{\kappa}_n = 2$. For $\widehat{\kappa}_n = 0$, it is easy to see that the above formulae reduce to those for penalized likelihood estimators, as given in \citet[Chapter 5]{Green:1994} highlighting yet again that density power divergence estimation is an extension of maximum likelihood estimation.

\subsection{Tuning and smoothing parameter selection}

As discussed previously, after the selection of $\Theta$, the estimator requires the choice of $\widehat{\kappa}_n$, $K_n$ and $\widehat{\lambda}_n$ which correspond to the tuning parameter, the dimensional of the approximating subspace and the penalty parameter, respectively. Among these, $K_n$ seems to be the least critical for the success of the estimator, as extensive experience with penalized lower rank estimators has shown that provided that the approximating subspace is rich enough, but still smaller than the sample size $n$, the dimension makes little difference. In our experience, a choice such as $K = [\min\{30, n/4\}]$, which ensures at least $4$ observations per basis function and puts a cap at $30$ basis functions, is appropriate for many situations. The number of basis functions can be increased beyond 30 in highly complex situations, but these tend to be rather rare in practice. Hence, the rest of this subsection is dedicated to the selection of $\widehat{\kappa}_n$ and $\widehat{\lambda}_n$.

The parameter $\widehat{\kappa}_n$ is particularly important as it determines the trade-off between robustness and efficiency of our penalized power divergence estimators and as such it is important to select it properly. Similar to \citet{Kal:2022}, our strategy of selecting $\widehat{\kappa}_n$ amounts to minimizing an approximation of the mean-integrated squared error (MISE) $\mathbb{E}\{\|\widehat{\beta}_n - \beta_0\|^2\}$. To clearly show the dependence of $\widehat{\beta}_n$ on $\kappa$ we herein denote it with $\widehat{\beta}_{n,\kappa}$. Standard calculations show that for each $\kappa \geq 0$ the MISE can be decomposed as
\begin{align}
\label{eq:5}
\MISE(\kappa) = \left\|\mathbb{E}\{\widehat{\beta}_{n,\kappa}\} - \beta_0 \right\|^2 + \mathbb{E}\left\{\left\|\widehat{\beta}_{n,\kappa} - \mathbb{E}\{\widehat{\beta}_{n,\kappa}\}  \right\|^2 \right\},
\end{align}
where the first term on the RHS represents the integrated squared bias and the second term represents the integrated variance of $\widehat{\beta}_{n,\kappa}$. None of these terms can be computed explicitly, as we lack knowledge of both $\beta_0$ and $\mathbb{E}\{\widehat{\beta}_{n,\kappa}\}$ (recall that $\widehat{\beta}_{n,\kappa}$ is defined implicitly as the solution of \eqref{eq:4}). To approximate the bias term of $\MISE(\kappa)$, we follow the strategy of \citet{War:2005} and replace $\mathbb{E}\{\widehat{\beta}_{n,\kappa}\}$ with its unbiased estimator $\widehat{\beta}_{n,\kappa}$. Moreover, we use a "pilot" estimator in the place of the unknown $\beta_0$. We have found in our numerical experiments that that any density power divergence estimator with a large enough tuning serves is an effective "pilot" estimator in practice, as it is less susceptible to outlying observations. Hence we use $\widehat{\beta}_{n,2}$ as our initial estimator for $\beta_0$.

To approximate the variance term, observe first that $\mathbb{E}\{\|\widehat{\beta}_{n,\kappa} - \mathbb{E}\{\widehat{\beta}_{n,\kappa}\}  \|^2 \} = \Tr\{ \mathbf{P}_0 \Cov\{\widehat{\boldsymbol{\beta}}_{n, \kappa}\} \}$ with $\mathbf{P}_0$ consisting of $\langle \theta_{i}, \theta_j \rangle, i,j = \ldots, K_n$. Let us next write
\begin{align*}
l_{\kappa}^{\prime}(Y_i, p_i(\alpha,\beta)) = \frac{\partial l_{\kappa}(Y_i, x) }{\partial x}  \bigg \vert_{x = p_i(\alpha, \beta)} \frac{\partial H(x)}{\partial x} \bigg \vert_{x = \eta_i} , \quad (i=1, \ldots, n),
\end{align*}
for the first derivatives and similarly $l_{\kappa}^{\prime \prime}(Y_i, p_i(\alpha,\beta))$ for the second derivatives. Define $n \times n$ matrices $\mathbf{C}_{\kappa} = \diag\{|l^{\prime}_{\kappa}(Y_i,p_i(\widehat{\alpha}_{n,\kappa}, \widehat{\beta}_{n, \kappa}))|^2\}$ and $\mathbf{D}_{\kappa}  = \diag\{l^{\prime \prime}_{\kappa}(Y_i,p_i(\widehat{\alpha}_{n,\kappa}, \widehat{\beta}_{n, \kappa}) \}$. A first order Taylor expansion of the score function of \eqref{eq:4} yields the approximation
\begin{align*}
\Cov\{(\widehat{\alpha}_{n,\kappa},\widehat{\boldsymbol{\beta}}_{n, \kappa})\} \approx \left[ \mathbf{B}^{\star \top} \mathbf{D}_{\kappa} \mathbf{B}^{\star} + 2 \widehat{\lambda}_{n, \kappa} \mathbf{P} \right]^{-1} \mathbf{B}^{\star \top} \mathbf{C}_{\kappa} \mathbf{B}^{\star} \left[ \mathbf{B}^{\star \top} \mathbf{D}_{\kappa} \mathbf{B}^{\star} + 2 \widehat{\lambda}_{n, \kappa} \mathbf{P}^{\star} \right]^{-1},
\end{align*}
where $\mathbf{B}^{\star}$ has a column of ones and $(i,j)$th element $\langle X_i, \theta_j \rangle, i = 1,\ldots, n, j = 2\ldots, (K_n+1)$ and $\mathbf{P}^{\star}$ is the penalty matrix $\mathbf{P}$ augmented with a row and a column of zeroes. An approximation to $\Cov\{\widehat{\boldsymbol{\beta}}_{n, \kappa}\} \}$ may be subsequently obtained by selecting the corresponding elements from the larger covariance matrix $\Cov\{(\widehat{\alpha}_{n,\kappa},\widehat{\boldsymbol{\beta}}_{n, \kappa})\}$.

With the foregoing approximations of the bias and variance we obtain the approximate mean-squared error (AMISE) given by
\begin{align*}
\AMISE_{2}(\kappa) & =  \left(\widehat{\boldsymbol{\beta}}_{n,\kappa} - \widehat{\boldsymbol{\beta}}_{n,2}\right)^{\top} \mathbf{P}_0 \left(\widehat{\boldsymbol{\beta}}_{n,\kappa} - \widehat{\boldsymbol{\beta}}_{n,2}\right) + \Tr\left\{\mathbf{P}_0\Cov\{\widehat{\boldsymbol{\beta}}_{n, \kappa}\} \right\},
\end{align*}
where the subscript $2$ indicates that $\widehat{\beta}_{n,2}$ is used in place of the unknown $\beta_0$. As a first step, we select $\widehat{\kappa}_n$ by minimizing $\AMISE_1(\kappa)$ over a grid of $20$ equidistant values in $[0,2]$. This grid includes both penalized maximum likelihood estimators ($\kappa = 0$) as well as more robust estimators ($\kappa  > 0$). Both bias approximation and the subsequent selection of $\kappa$ depends on the pilot estimator, $\widehat{\beta}_{n,2}$. To reduce this dependence we further follow the proposal of \citet{Basak:2021} and iterate the selection procedure. That is, after identifying the value of $\kappa$ minimizing $\AMISE_{2}(\kappa)$, $\widehat{\kappa}_n$, we use $\widehat{\beta}_{n,\widehat{\kappa}_n}$ as the new pilot estimator and minimize $\AMISE_{\widehat{\kappa}_n}(\kappa)$.  This process is repeated until $\widehat{\kappa}_n$ converges.  This procedure requires merely the computation of a number of estimators $(\widehat{\alpha}_{n,\kappa}, \widehat{\beta}_{n,\kappa})$ for different values of $\kappa$ and, thanks to the rapid convergence of our modified Fisher-scoring algorithm, the associated computational burden is minimal.

The computation of the estimator for a given value of $\kappa$ requires an appropriate value of the penalty parameter $\lambda$. To determine $\lambda$, we propose an appropriate modification of the AIC criterion adopted in \citet{Kal:2022}, viz,
\begin{align*}
\AIC(\lambda) = 2 \sum_{i=1}^n l_{\kappa}(Y_i, p_i(\widehat{\alpha}_{n, \kappa}, \widehat{\beta}_{n,k})) + 2 \Tr\left\{ \left[ \mathbf{B}^{\star \top} \mathbf{D}_{\kappa} \mathbf{B}^{\star} + 2 \lambda \mathbf{P}^{\star} \right]^{-1} \mathbf{B}^{\star \top} \mathbf{D}_{\kappa} \mathbf{B}^{\star} \right\}.
\end{align*}
Implementations and illustrative examples of the functional density power divergence in the \textsf{R} programming language \citep{R}  are available at \url{https://github.com/ioanniskalogridis/Robust-and-Adaptive-Functional-Logistic-Regression}.

\section{Numerical examples}
\label{sec:5}

We now examine the practical performance of the proposed penalized density power divergence via a numerical study. In our numerical experiments we are primarily interested in assessing how well the coefficient function $\beta_0$ in a variety of settings that include local behaviour of $\beta_0$, noisy and outlying predictors $X_i$ and misclassified responses $Y_i$. For this comparison, we consider the following estimators:
\begin{itemize}
\item The adaptive density power divergence estimator discussed in the previous section, denoted by DPD($\widehat{\kappa}_n$).
\item The adaptive density power divergence estimator corresponding to $\kappa = 0$, abbreviated as ML.
\item The density power divergence estimator with fixed $\kappa = 1$, abbreviated as DPD($1$)
\item The density power divergence estimator with fixed $\kappa = 2$, abbreviated as DPD($2$).
\end{itemize}
Recall that the adaptive estimator, DPD($\widehat{\kappa}_n)$, uses DPD($2$) as an initial guess for $\beta_0$ and the grid of candidate values for $\widehat{\kappa}_n$ is contained in $[0,2]$.
It is therefore of interest to compare how DPD($\widehat{\kappa}_n)$ performs against the estimators ML, DPD($1$) and DPD($2$) in practice. For all estimators we have used the popular cubic B-spline basis with $\mathcal{J}(f) = \|f^{\prime  \prime}\|^2$. Moreover, we have used the canonical logit link function in all subsequent computations. 

We have generated the predictor curves according to the truncated Karhunen-Loève decomposition given by
\begin{align}
X(t) = \sum_{j=1}^{50} j^{-1} Z_j \sqrt{2} \sin((j-1/2)\pi t), \quad t \in [0,1],
\label{curve_design}
\end{align}
where the $Z_j$ are independent standard Gaussian random variables. To deal with these curves practically we have discretized them in 200 equidistant points $t_j$ within the $[0,1]$-interval and used Riemann sums to compute inner products and norms whenever required. The predictor variables were combined  with each of the following three coefficient functions:
\begin{enumerate}
\item $\beta_1(t) = 3(t-0.3)^2+1$
\item $\beta_2(t) = 3 \sin(3.4 t^2) $
\item $\beta_3(t) = -\sin(5t/1.2)/0.5-1$.
\end{enumerate}
The binary responses $Y_i$ were then generated from a Bernoulli distribution with probability of "success" equal to $H( \langle X, \beta \rangle)$ where $H(x) = e^{x}/(1+e^{x})$. The coefficient functions are increasingly complex, as $\beta_2$ possesses more local characteristics than $\beta_1$ and, likewise, $\beta_3(t)$ exhibits more local behaviour than $\beta_2(t)$. 

\begin{figure}[H]
\centering
\subfloat{\includegraphics[width = 0.495\textwidth]{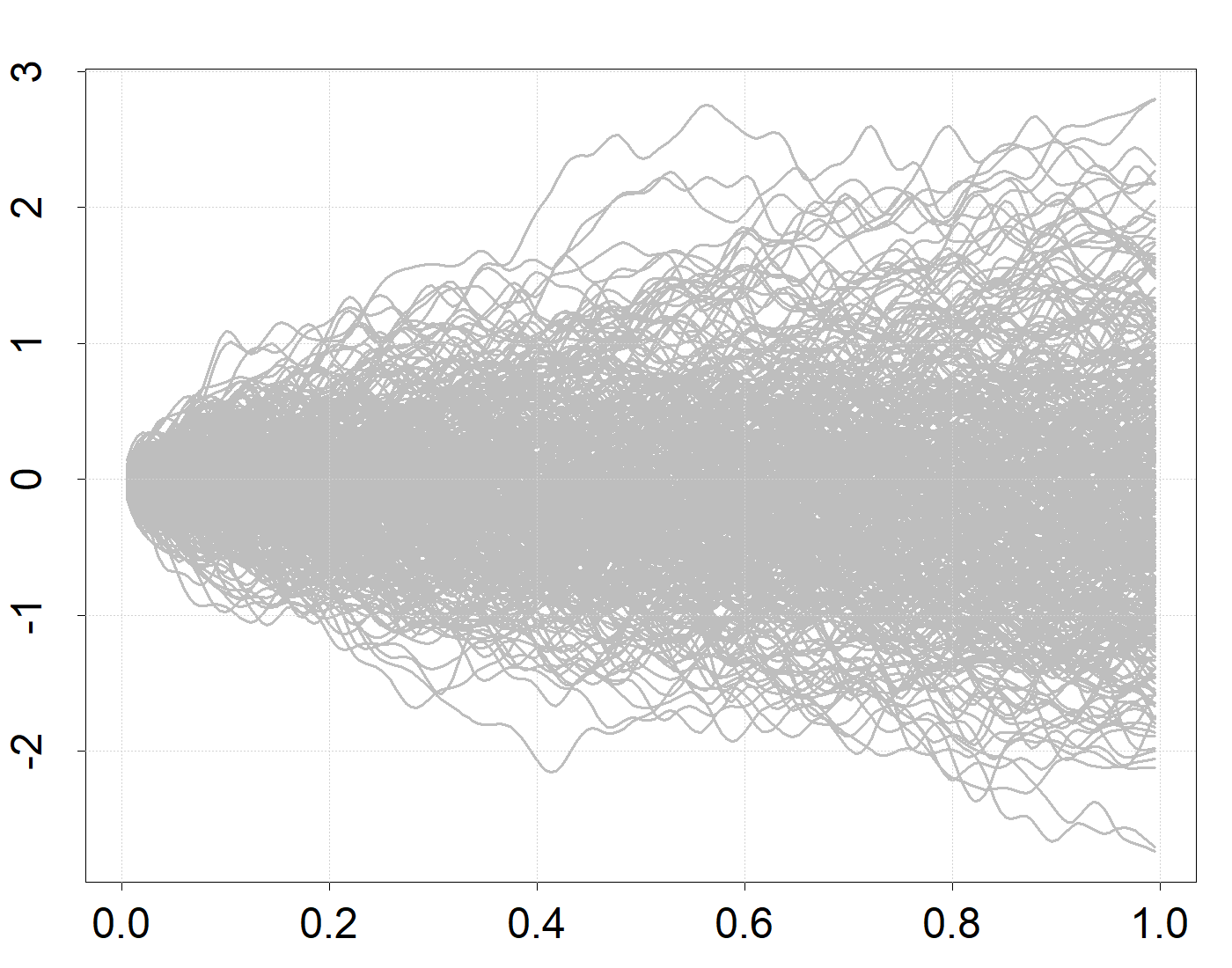}} \ 
\subfloat{\includegraphics[width = 0.495\textwidth]{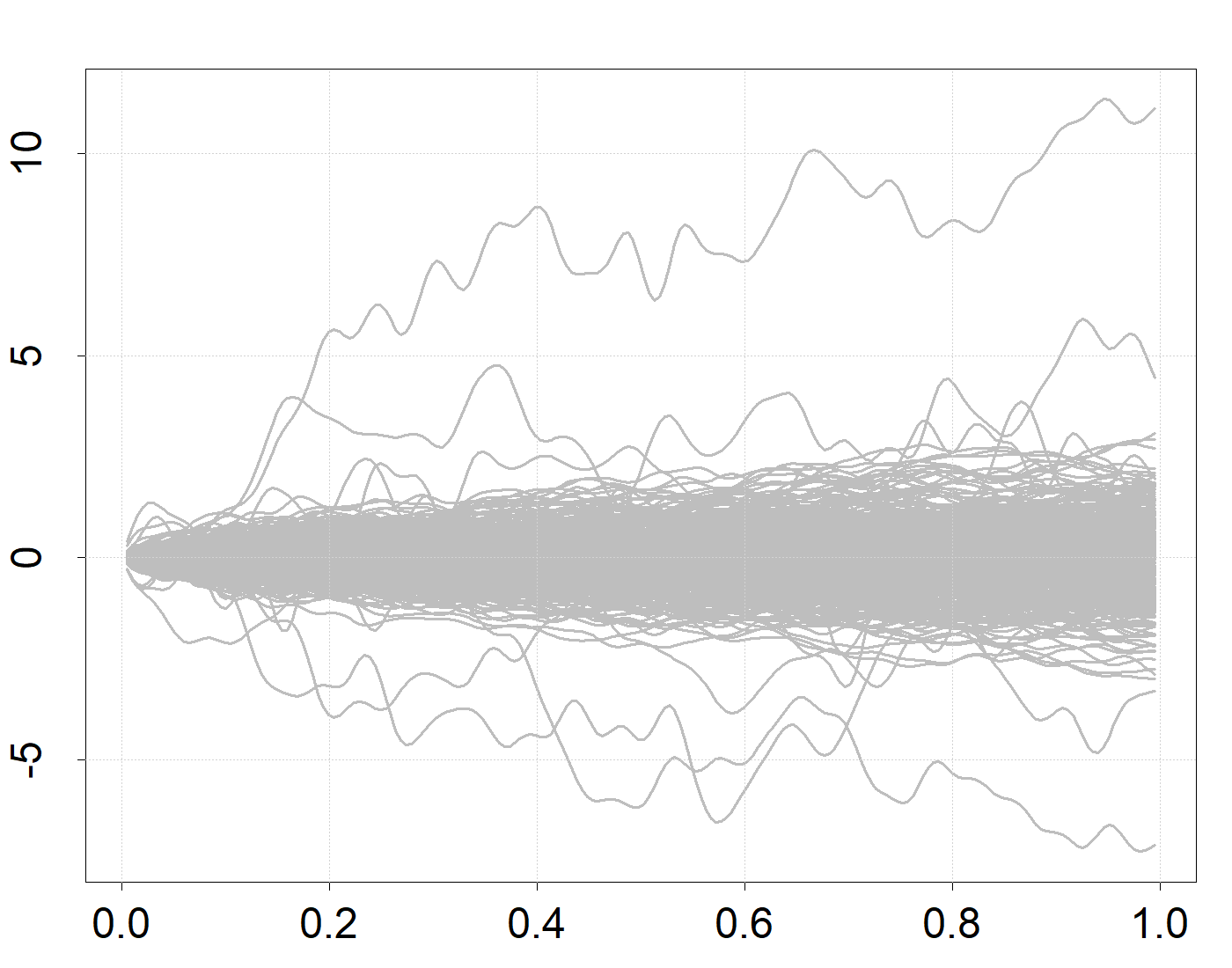}}
\caption{Two representative samples of simulated curves from clean and contaminated $(\epsilon = 0.02)$ data on the left and right respectively.}
\label{fig:1}
\end{figure}

Next to the clean data generation described above, we have also perturbed a number of observations $(X_i, Y_i)$ in the following manner. For $\epsilon \in \{0.01, 0.02, 0.05 \}$ we have replaced an $\epsilon$-fraction of the $X_i$ with $5X_i$ and we have also replaced the corresponding $Y_i$ with $1-Y_i$ thereby changing an $\epsilon$-fraction of the responses from $0$ to $1$ and vice versa. This contamination scenario corresponds to the introduction of bad leverage points into the data and it can be expected a priori that the performance of all estimators will be affected by these outlying observations. For a better understanding of the effect of the multiplication by $5$ on the shape and amplitude of the curves, Figure~\ref{fig:1} presents two representative samples of curves with $\epsilon = 0$ (clean data) and $\epsilon = 0.02$ (mild contamination).

For each coefficient function $\beta$ and level of contamination $\epsilon$, we evaluate the performance of the estimators according to the mean squared error given by
\begin{align*}
\MSE = \frac{1}{200} \sum_{j=1}^{200} \left|\widehat{\beta}_n(t_j) - \beta(t_j)\right|^2,
\end{align*}
which is an approximation to the $\mathcal{L}^2([0,1])$-distance. Since the MSE is bounded below by zero but unbounded above, its distribution tends to be right skewed and heavy-tailed. For this reason we have found that the median is an overall better measure of centrality than the mean. Therefore, Table~\ref{tab:1} below reports median MSEs (Median) for the four competing estimators along with their bootstrapped standard errors (BSE) based on 1000 simulated datasets with $n = 400$. 

\begin{table}[h!]
\centering
\resizebox{\columnwidth}{!}{%
\begin{tabular}{c c c c c c c c c c c } 
\multicolumn{1}{c}{}  & &  \multicolumn{2}{c}{DPD($\widehat{\kappa}_n)$} & \multicolumn{2}{c}{ML} & \multicolumn{2}{c}{DPD(1)} & \multicolumn{2}{c}{DPD(2)}   \\ \\[-2ex]
 & $\epsilon$ &  Median & BSE & Median & BSE & Median & BSE & Median & BSE \\ \\ [-1.5ex]
\multirow{4}{*}{$\beta_1$} & 0 & 0.389 &  0.024 & \textbf{0.385}  & 0.031 &   0.482 & 0.032  & 0.401 & 0.021 \\
& 0.01 & \textbf{0.504}  & 0.027 & 2.398 & 0.337 & 0.726 & 0.056 & 0.506  &  0.025    \\
& 0.02  & \textbf{0.714} &  0.041 & 8.654  & 1.179 & 1.263 & 0.095 & 0.723   &  0.047  \\
& 0.05  & \textbf{2.313} & 0.139 & 6.670 & 0.888 & 5.327 & 0.284 & 2.322 & 0.142
  \\  \\[-1.5ex]
\multirow{4}{*}{$\beta_2$} & 0  & 1.180 & 0.021  & \textbf{1.176} & 0.025  &  1.307 & 0.021 & 1.195 & 0.019 \\
& 0.01 &  \textbf{1.347}  & 0.029 & 5.292  &  0.810  &  1.642 &   0.052 & 1.358  & 0.034  \\
& 0.02  & \textbf{1.487} & 0.040 &   20.19 & 3.736  &  1.932 & 0.098 & 1.505 & 0.040   \\
& 0.05  & 2.929 & 0.129 & 30.69 & 3.925 & 14.38 & 1.908 & \textbf{2.927} & 0.117
 \\  \\[-1.5ex]
\multirow{4}{*}{$\beta_3$} & 0 & \textbf{1.038} &  0.027 & 1.040  & 0.023 &   1.160 & 0.038 & 1.062 & 0.025 \\
& 0.01 &  \textbf{1.064}  & 0.029 & 2.520  & 0.266 & 1.196 & 0.045 & 1.096  & 0.024    \\
& 0.02  & \textbf{1.278} &  0.039 & 10.43  & 1.373 & 1.797 & 0.095 & 1.302  & 0.040  \\
& 0.05  & \textbf{2.915} & 0.180 & 11.15 & 1.617 &  8.473 & 2.057 & 2.946 & 0.195
\end{tabular}}
\caption{Median and boostrapped standard errors (obtained from 10000 bootstrap replications) of the mean-squared errors for the competing estimators over 1000 datasets of size $n=400$. Best median performances are in bold.}
\label{tab:1}
\end{table}

There are several interesting observations emerging from Table~\ref{tab:1}. It may be immediately seen that while penalized maximum likelihood estimators can be effective in estimating the true coefficient function in the absence of contamination, their performance greatly deteriorates even with even a small fraction of outliers. In particular, even for the smallest level of contamination the performance of the penalized ML estimator decreases 2 or even 3-fold. The lack of resistance of ML based estimators needs to be contrasted with the stability of density power divergence estimators with positive tuning parameters. The performance of these estimators almost matches the performance of ML based estimators in clean data and barely deteriorates for $\epsilon \in \{0.01, 0.02\}$. Higher levels of contamination do take their toll on the performance of density power divergence estimators, but clearly not to the same extent as with ML based estimators. As a means of visualizing the difference in practical performance of ML-based and density power divergence estimators, Figure~\ref{fig:2} presents the 1000 estimates for $\beta_1(t)$ of the ML and DPD($\widehat{\kappa}_n)$ procedures in clean data as well as in mildly contaminated data on the first and second rows respectively.

\begin{figure}[H]
\centering
\subfloat{\includegraphics[width = 0.495\textwidth]{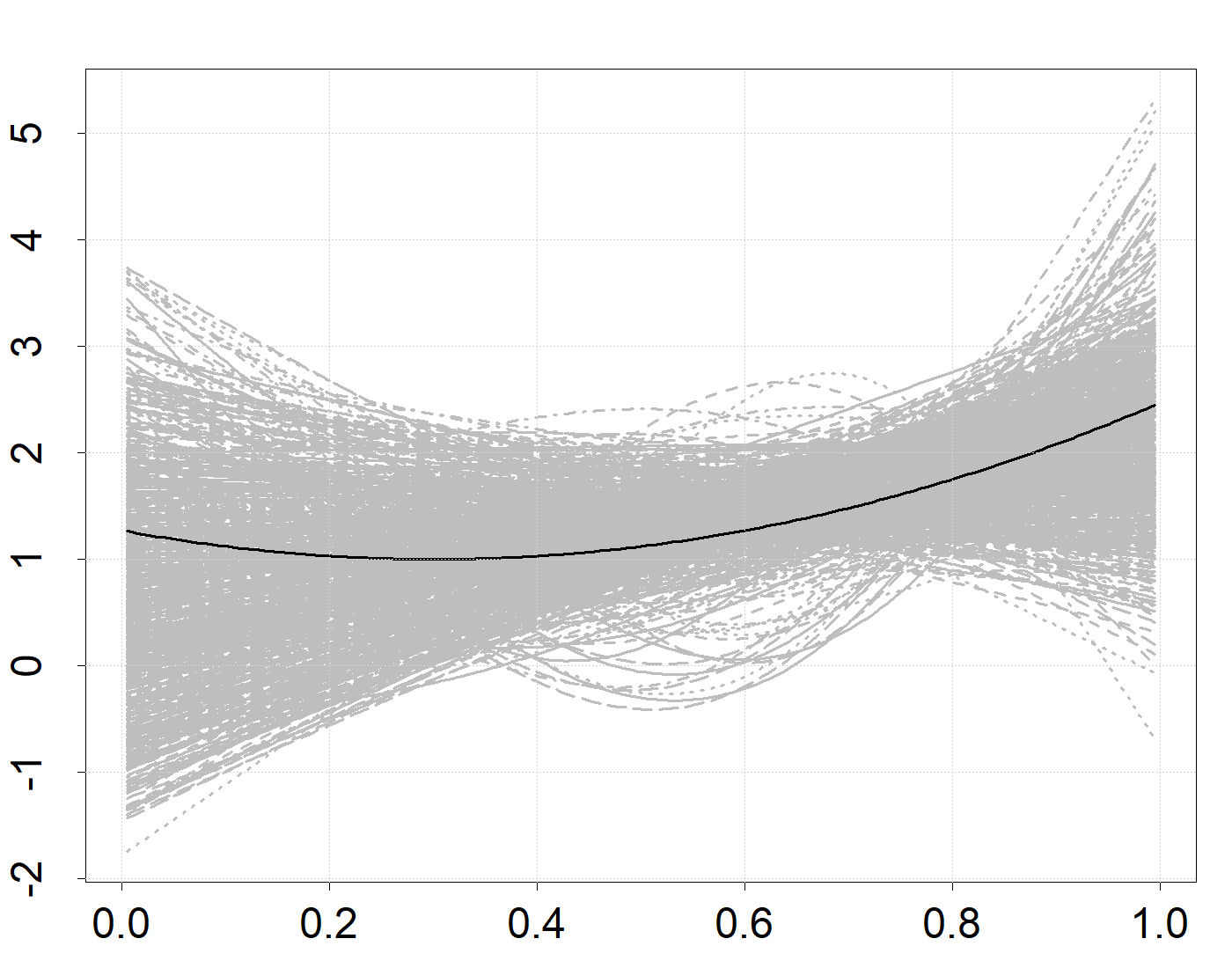}} \ 
\subfloat{\includegraphics[width = 0.495\textwidth]{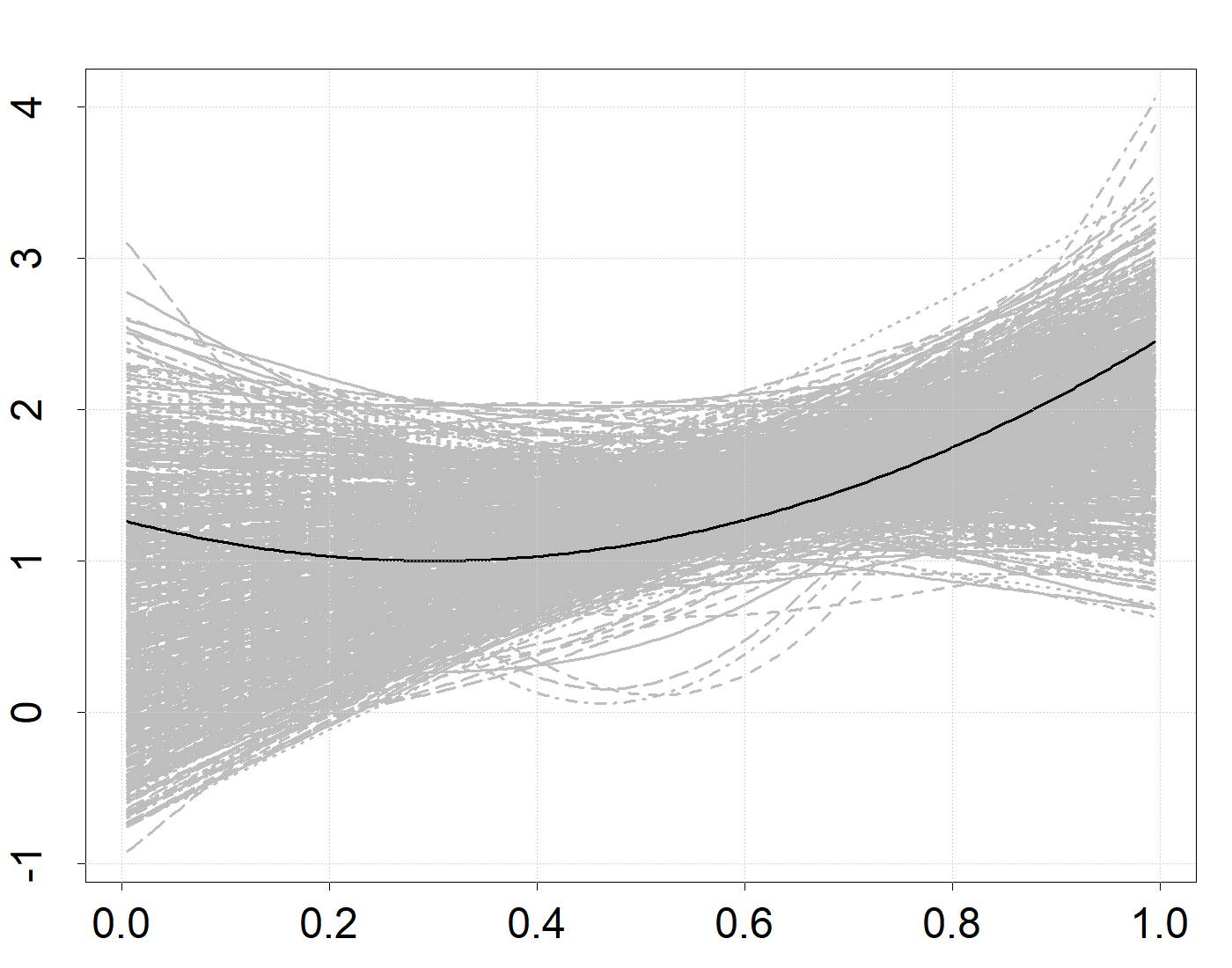}} \\
\subfloat{\includegraphics[width = 0.495\textwidth]{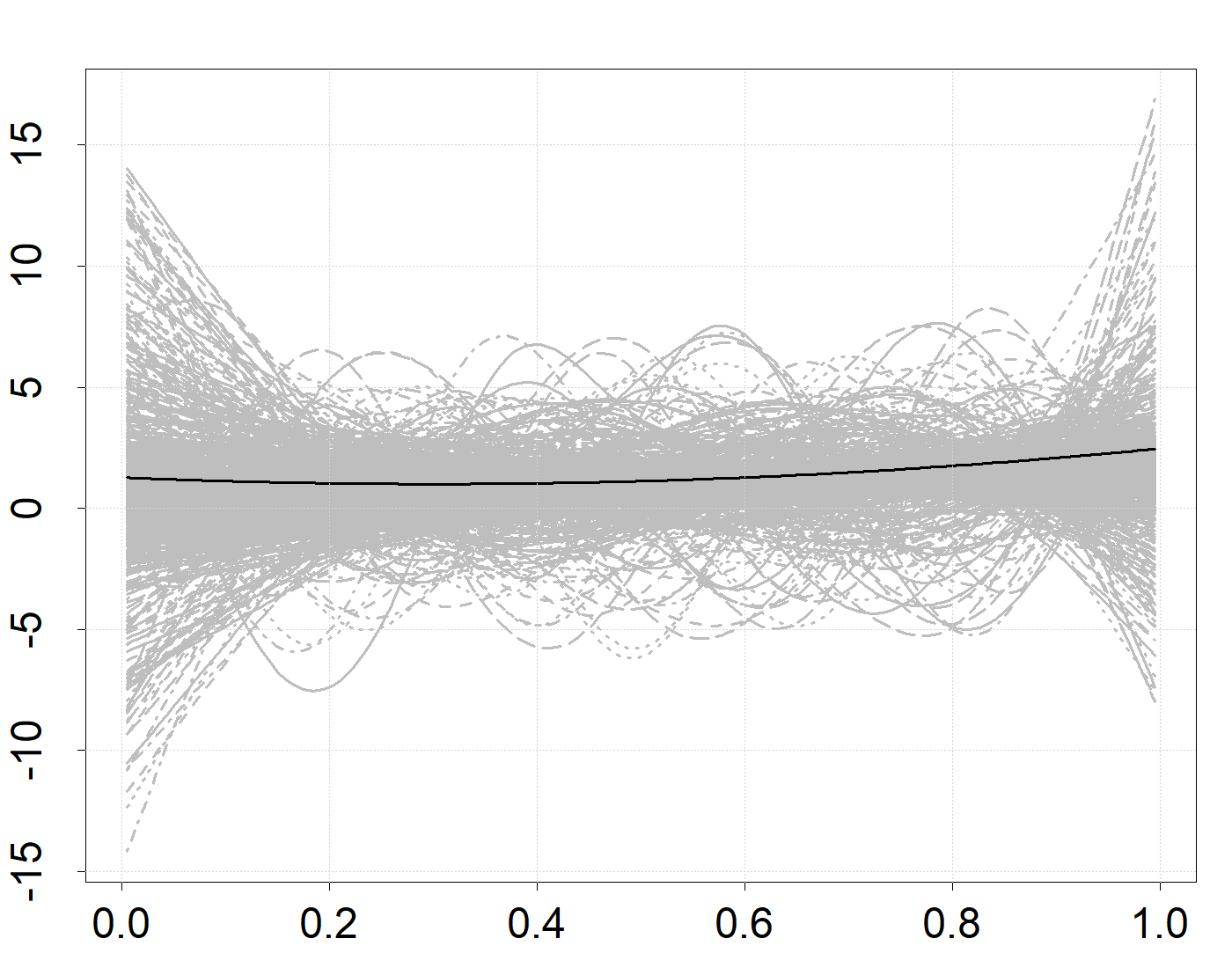}} \
\subfloat{\includegraphics[width = 0.495\textwidth]{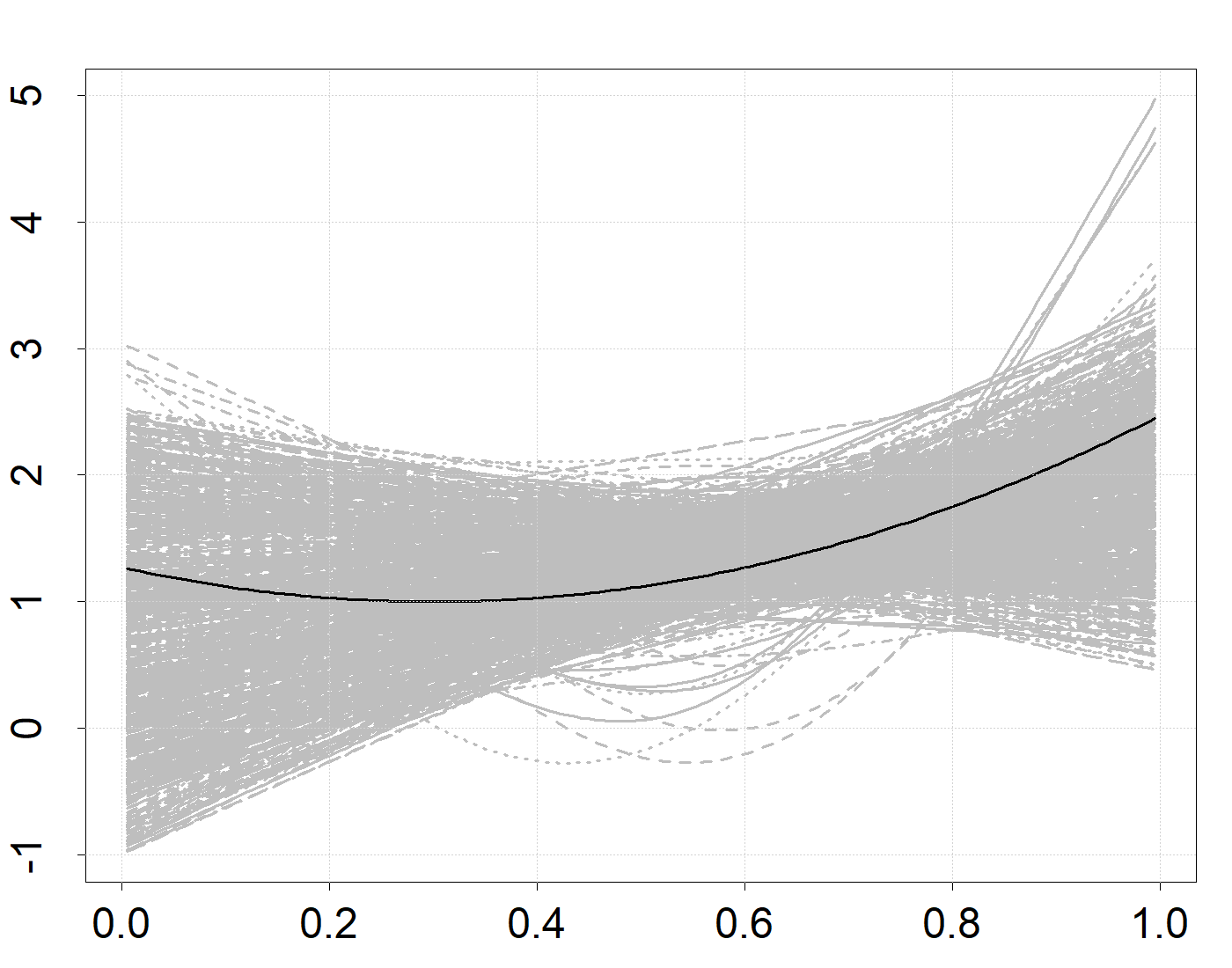}}
\caption{Top row: $1000$ ML and DPD($\widehat{\kappa}_n)$-estimates for $\beta_1(t)$ with $\epsilon = 0$ on the left and right panels respectively. Bottom row: $1000$ ML and DPD($\widehat{\kappa}_n)$-estimates for $\beta_1(t)$ with $\epsilon = 0.01$ on the left and right panels respectively. The solid black lines depict the true coefficient function $\beta_1(t)$.}
\label{fig:2}
\end{figure}

Comparing the robust estimators DPD($\widehat{\kappa}_n$), DPD(1) and DPD(2) in detail reveals that the adaptive nature of the DPD($\widehat{\kappa}_n$)-estimator allows for higher efficiency than their competitors in clean data without sacrificing robustness whenever outlying observations are present among the data. In fact, the performance of DPD($\widehat{\kappa}_n$) is similar to that of DPD(2) under high levels of contamination and substantially better than the performance of the less robust DPD(1)-estimator. The reason for this superior performance may be traced to the construction of DPD($\widehat{\kappa}_n)$-estimators, which ensures that the tuning parameter is selected in such a way that the mean-integrated squared error is approximately minimized. This criterion tends to yield a large value of $\kappa$ in the presence of contamination and a smaller value of $\kappa$ in clean data thereby ensuring the good all-round performance of DPD($\widehat{\kappa}_n$)-estimators.

\section{Application: Gait Analysis}
\label{sec:6}

It is frequently of interest to determine whether an individual walks normally or her walk is hindered due to an underlying reason, for example, pain in the abdomen or wounded toes. The present dataset consists of $166$ observations on individuals performing a walk cycle of a few seconds during which the position of the toes in the vertical axis is recorded via a motion sensor at $343$ points. These trajectories are depicted in the left panel of Figure~\ref{fig:RD1}. An indicator variable with values of $1$ and $0$ depending on whether the walk is normal or not accompanies these trajectories. The goal of this study is to determine the probability of a normal walk based on the position of the toes during this cycle. The dataset is freely available under the name "ToeSegmentation2" in the UCR Time Series Classification Archive \citep{UCR:2018}.

\begin{figure}[H]
\centering
\subfloat{\includegraphics[width = 0.495\textwidth]{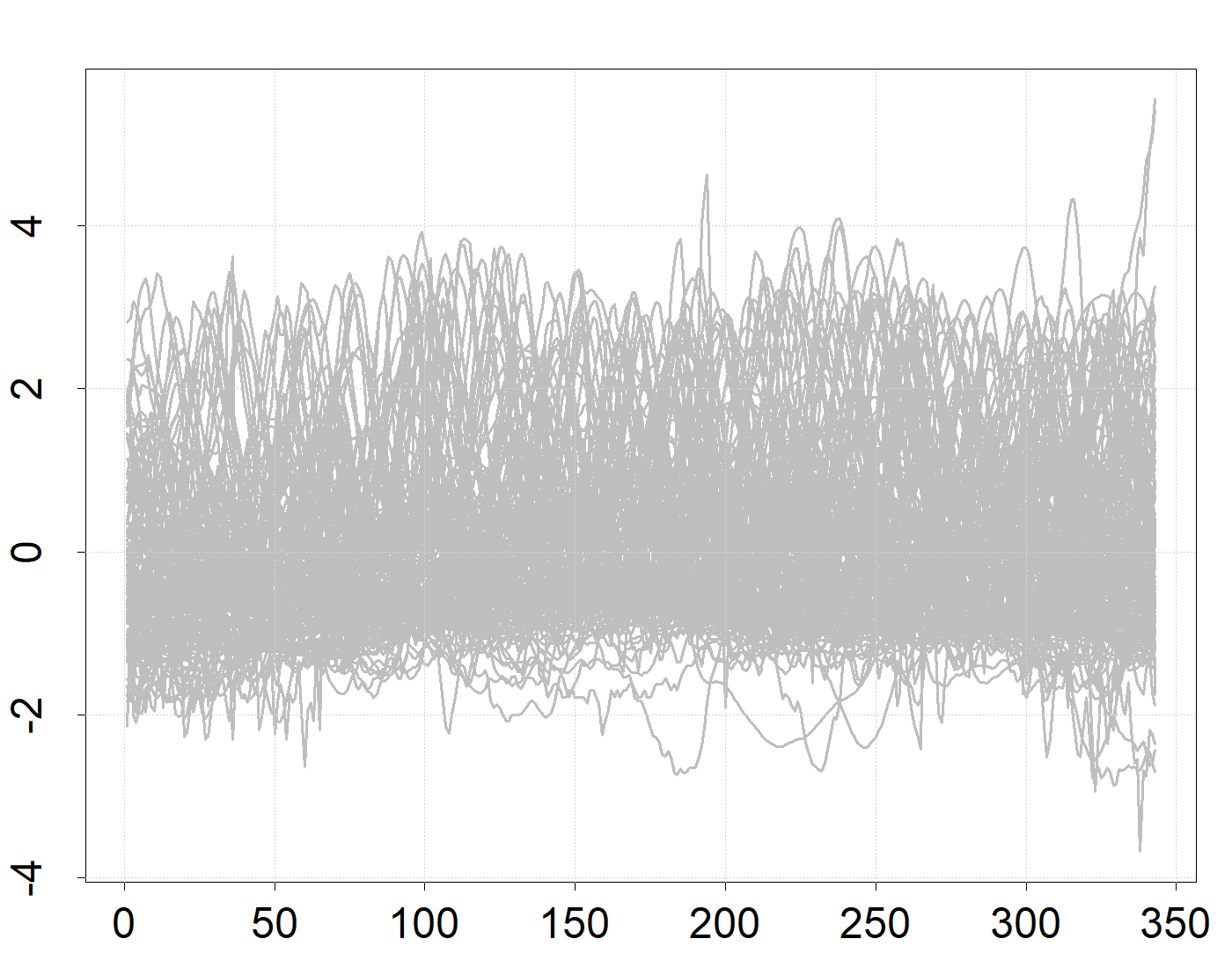}} \ 
\subfloat{\includegraphics[width = 0.495\textwidth]{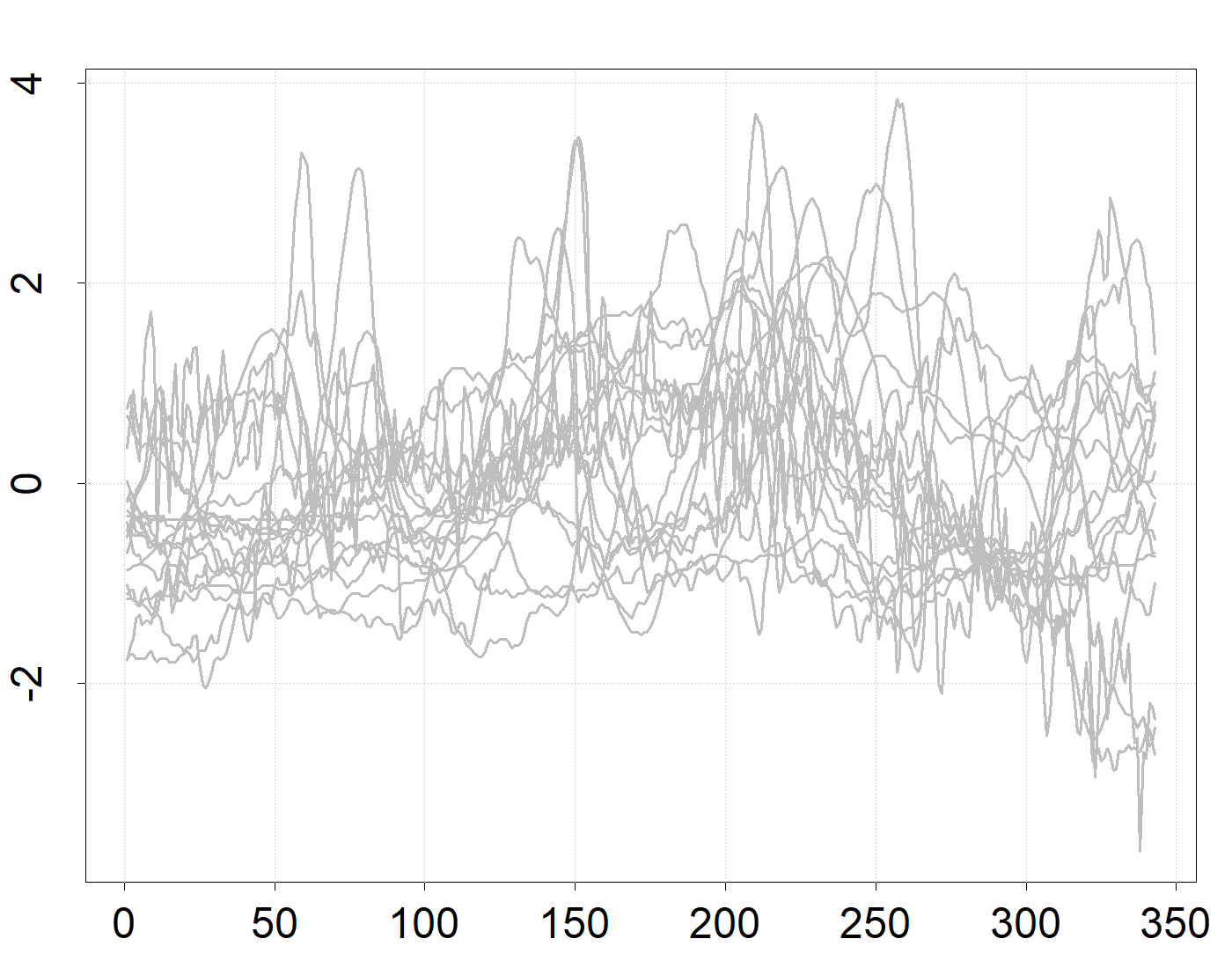}}
\caption{Left: Position of the toes of $166$ subjects in the vertical axis over a gait cycle. Right: Trajectories of observations with large Anscombe residuals according to the DPD($\widehat{\kappa}_n)$-estimates.}
\label{fig:RD1}
\end{figure}

In most real-data situations, physiological differences among individuals can be expected to lead to a number of outlying observations in the form of vertical outliers, that is, observations that do not quite follow the model. In our case, these may be either individuals that exhibit a normal walking pattern yet suffer from pain or discomfort or individuals whose walk is abnormal despite the fact that they are otherwise healthy. In order to demonstrate the effects of such observations on non-robust estimators, we have computed both the penalized B-spline based maximum likelihood estimator proposed by \citet{Cardot:2005} with the canonical logit link and the analogous adaptive density power divergence estimator.

The coefficient function estimates are depicted in the left panel of Figure~\ref{fig:RD2} in solid blue and dashed red lines respectively. The panel immediately reveals a striking difference between these two estimates, in that while the DPD($\widehat{\kappa}_n)$-estimates exhibit a number of peaks, most notably around the 100th and 300th points within the cycle, the maximum likelihood estimates form what is essentially a straight line, which indicates severe oversmoothing. As a result of this oversmoothing, the maximum likelihood estimator is unable to identify the points of the cycle that are influential in the determination of whether one's walking is normal or abnormal. In view of our prior discussion, it may be conjectured that a number of outlying observations have exerted undue influence on both the maximum likelihood estimates themselves and the selected penalty parameter, which in this case is carried out in a non-resistant manner.

In order to ascertain whether some observations are indeed influential on the ML estimates, we may look at the residuals produced by the resistant DPD($\widehat{\kappa}_n)$-estimates for which the algorithm described in Section~\ref{sec:4} selects $\widehat{\kappa}_n = 1.9$ indicating the presence of several outlying observations. To detect these observations we may look at the Anscombe residuals \citep[p. 29]{MCN:1983}, which for the Bernoulli distribution are given by
\begin{align*}
r_{A,i} = \frac{\IB(Y_i, 2/3, 2/3) -\IB(\widehat{\mu}_i, 2/3, 2/3) }{\widehat{\mu}_i^{1/6}(1-\widehat{\mu}_i)^{1/6}}, \quad (i = 1, \ldots, n),
\end{align*}
where $\IB(x,a,b) = \int_{0}^x t^{a-1}(1-t)^{b-1} dt$ and $\widehat{\mu}_i$ is $i$th fitted value, i.e., $\widehat{\mu}_i = H(\widehat{\alpha}_n+ \int_{1}^{343} X_i(t) \widehat{\beta}_n(t)  dt)$. The Anscombe residuals follow the standard Gaussian distribution much more closely than their Pearson counterparts, hence we may apply conventional outlying detection rules. In particular, we may classify an observation as an outlier if $|r_{A,i}| \geq 2$. With this rule, the DPD($\widehat{\kappa}_n$)-estimates point towards 19 outlying observations, $13$ of which have Anscombe residuals larger than $10$ and can thus be considered extreme outliers. The trajectories of the $19$ detected outlying observations are plotted in the right panel of Figure~\ref{fig:RD2}. Interestingly, $18$ of these observations are classified as normal walk which indicates that in the vast majority of cases it is individuals with underlying issues that are misclassified as walking normally and not the other way around.

As a sensitivity check, one may decide to remove the detected outlying observations from our data and recompute the estimates. Doing so results in the revised estimates presented in the right panel of Figure~\ref{fig:RD2}. The panel reveals that, while the DPD($\widehat{\kappa}_n$)-estimates have undergone minimal adjustment, the ML-estimates have markedly changed and are now completely identical with the DPD($\widehat{\kappa}_n$)-estimates. We may thus conclude that the presence of outliers have negatively impacted the non-resistant ML estimates, but not the much more resistant power density divergence estimates, which yield reliable estimates both in the presence and absence of outlying observations. Furthermore, as these estimates are much less drawn towards outlying observations, they can also be used for the purpose of outlier detection allowing us to identify influential observations that would have been otherwise missed by ML-based estimates.

\begin{figure}[H]
\centering
\subfloat{\includegraphics[width = 0.495\textwidth]{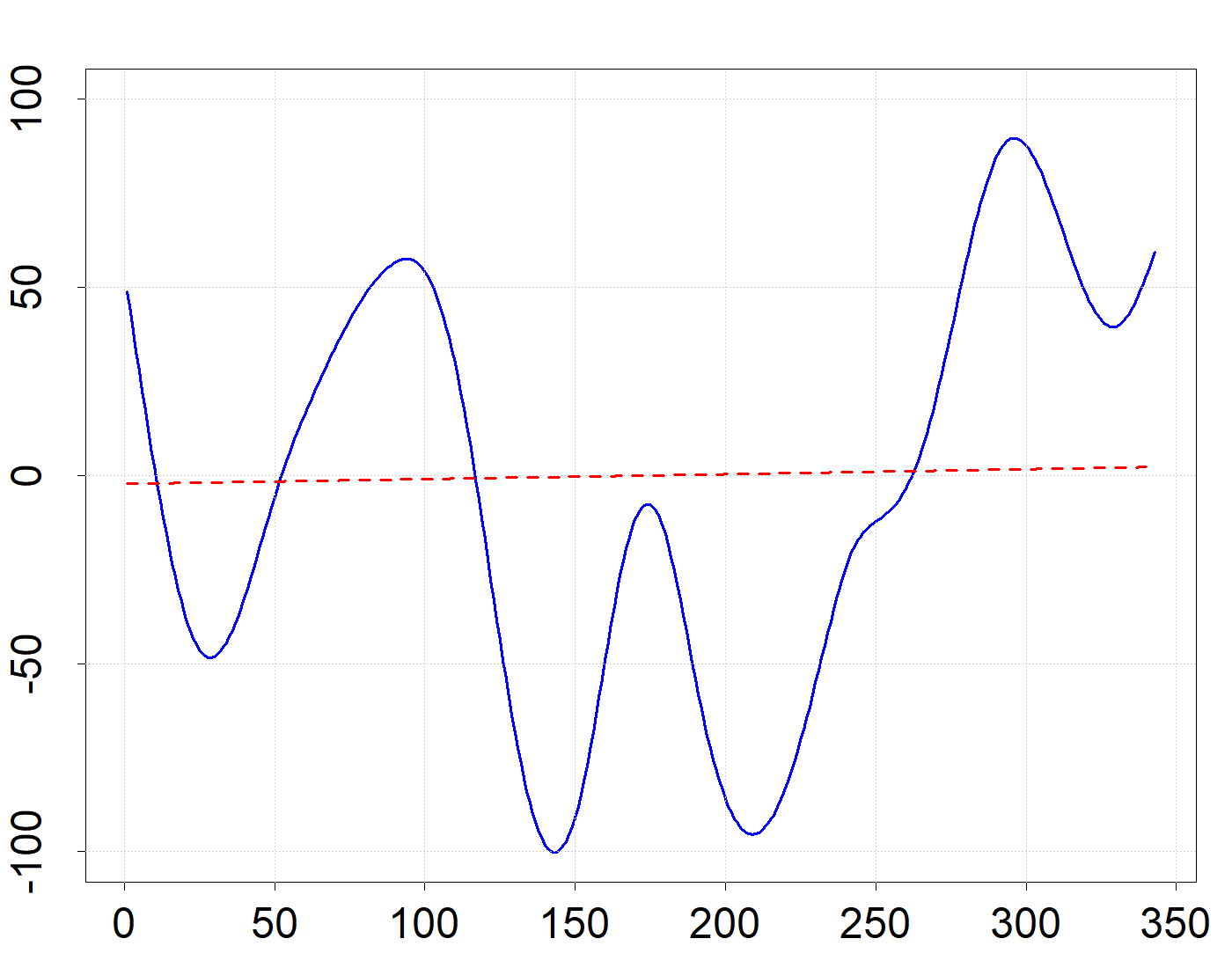}} \ 
\subfloat{\includegraphics[width = 0.495\textwidth]{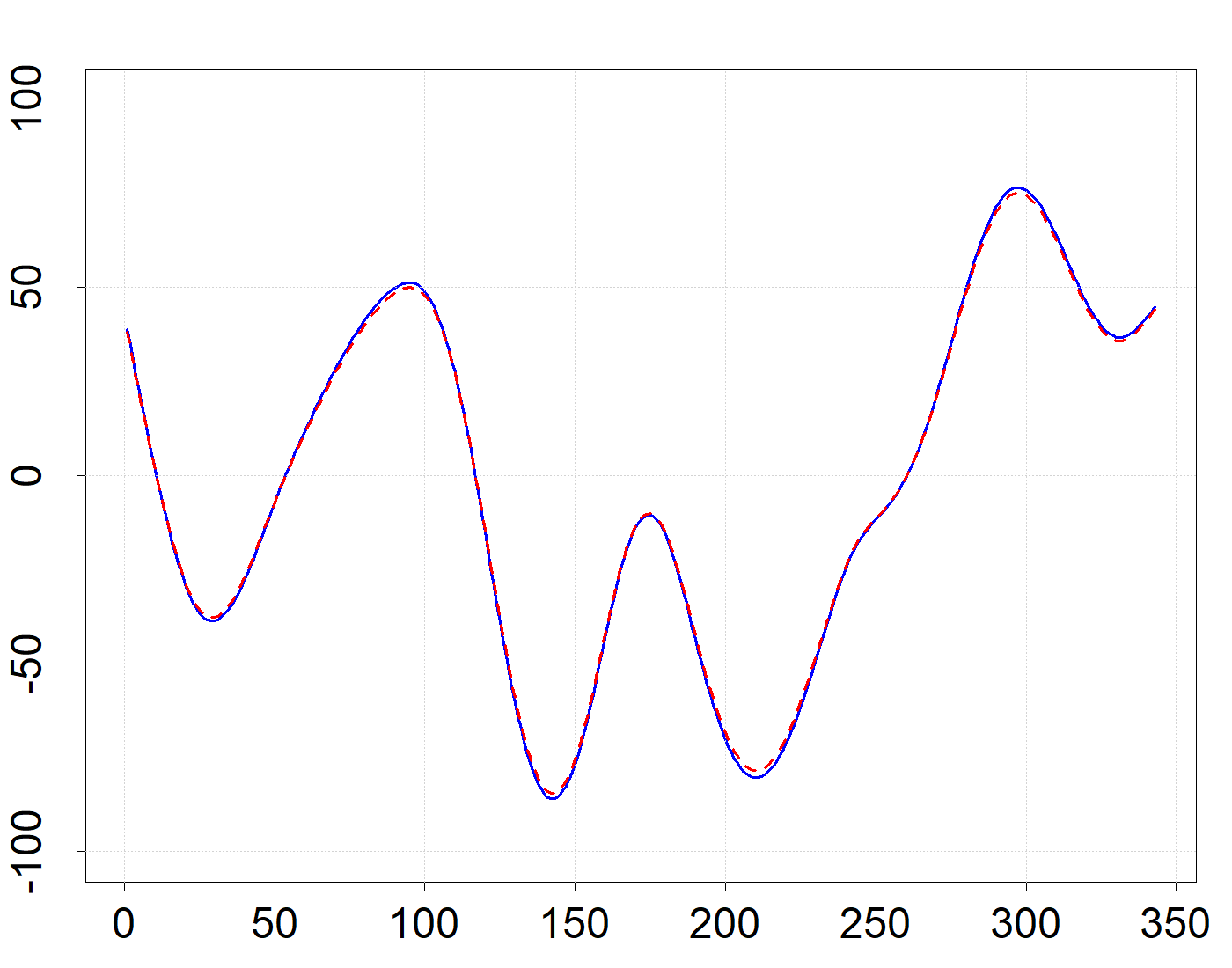}}
\caption{Left: Coefficient function estimates based on the full dataset from the DPD($\widehat{\kappa}_n$) and ML estimators in solid blue and dashed red lines respectively. Right: Coefficient function estimates obtained upon removing observations with large Anscombe residuals according to the DPD($\widehat{\kappa}_n$)-estimates.}
\label{fig:RD2}
\end{figure}

\section{Concluding remarks}
\label{sec:7}

This paper introduces a family of estimators for the functional logistic regression model that serves as a more balanced alternative to popular maximum-likelihood based estimators, whose lack of robustness makes them particularly vulnerable to outlying observations. The proposed family of estimators can achieve high efficiency in clean data and robustness against outlying observations through appropriate tuning and, what is more, we have proposed an algorithm that selects the tuning in an automatic fashion so that the applied researcher only needs to select the approximating subspace and penalty functional. The proposed estimators possess strong theoretical properties in that they are uniformly consistent and can achieve high rates of convergence with respect to the prediction error even with random tuning parameters. These properties make the proposed methodology particularly well-suited for the analysis of many complex datasets commonly encountered nowadays, such as the gait dataset of Section~\ref{sec:6}.

There is a number of interesting and practically useful extensions we aim to consider in future work, most notable of which is the generalization to other GLMs. While our methodology based on penalized density power divergence can be readily applied regardless of which distribution in the exponential family that the response follows, the theoretical properties established herein are specific to the Bernoulli distribution and do not immediately extend to other GLMs. These theoretical properties likely remain true under the right set of conditions and the investigation of such conditions is, in our opinion, a worthwhile endeavour. Our work can also be extended to more versatile distributions of the exponential family involving also a scale parameter, such as the general gamma and log-normal distributions which are frequently used to model claims in the insurance industry. Due to the frequent presence of outlying observations in such datasets, the use of robust methods is well-warranted and we are confident that penalized density power divergence will likewise lead to computationally feasible estimators with strong theoretical properties.

\section*{Acknowledgements}

The author gratefully acknowledges the support by the Research Foundation-Flanders (project 1221122N).

\section{Appendix: Proofs of the theoretical results}

\begin{lemma}[Bounds on the density power divergence]
\label{lem:1}
Let h and g denote Bernoulli densities with probabilities of "success" equal to $p_h$ and $p_g$ respectively. For every $\kappa_0 \in (0, 1]$ we have
\begin{align*}
0 \leq d_{\kappa_0}(h,f) \leq 1+ \frac{1}{\kappa_0},
\end{align*}
and the upper bound is attained if, and only if, $|p_f-p_h|=1$.

\end{lemma}
\begin{proof}
The lower bound follows by the arithmetic-geometric mean (AM-GM) inequality, as the densities are nonnegative and for every $y \in \{0,1\}$ we have
\begin{align*}
h(y)f^{\kappa_0}(y) = \left\{h^{1+\kappa_0}(y)\right\}^{\frac{1}{1+\kappa_0}}\left\{ f^{1+\kappa_0}(y)\right\}^{\frac{\kappa_0}{1+\kappa_0}} \leq \frac{1}{1+\kappa_0}h^{1+\kappa_0}(y) + \frac{\kappa_0}{1+\kappa_0}f^{1+\kappa_0}(y).
\end{align*}
The inequality is strict unless $h = f$ on $\{0,1\}$. Multiplying by $(\kappa_0+1)/\kappa_0 >0$ yields the result. For the upper bound, write
\begin{align*}
d_{\kappa_0}(h,f) = \sum_{y \in \{0,1\}} \{f(y) - h(y) \}f^{\kappa_0}(y) + \frac{1}{\kappa_0} \sum_{y \in \{0,1\}} \{h^{\kappa_0}(y) -f^{\kappa_0}(y) \}h(y),
\end{align*}
and observe that 
\begin{align*}
d_{\alpha_0}(h,f) & \leq \left|p_h-p_f\right|\left|(1-p_f)^{\kappa_0}-p_f^{\kappa_0} \right| + \frac{1}{\kappa_0} \sum_{y \in \{0,1\}} \left| h^{\kappa_0}(y) -f^{\kappa_0}(y) \right| h(y)
\\ & \leq   1 + \frac{1}{\kappa_0},
\end{align*}
by the reverse triangle inequality and the facts that $|p_h-p_f| \leq 1$ and $| h^{\kappa_0}(y) -f^{\kappa_0}(y) | \leq 1$. To verify that the upper bound is attainable for $|p_f-p_h| = 1$ notice that in that case it must be that either $p_f = 0$ and $p_h = 1$ or $p_f = 1$ and $p_h = 1$. Simple algebra shows that the above inequalities become in both cases equalities and the reverse implication is also true.

\end{proof}

\begin{lemma}[Fisher consistency]
\label{lem:2}
For any $\kappa_0>0$ define
\begin{align*}
M(\beta, \kappa_0) = \mathbb{E}\{ d_{\kappa_0}(f_{p_1(\beta_0)}, f_{p_1(\beta)})  \}, \quad g \in \mathcal{B}([0,1]),
\end{align*}
with $\beta_0$ denoting the true coefficient function. Under assumptions (A2) and (A5), $M(\beta, \kappa_0) \geq 0$ and $M(\beta,\kappa_0) = 0$ if, and only if, $\beta = \beta_0$.
\end{lemma}

\begin{proof}

Observe initially that, by Lemma~\ref{lem:1}, $M(\beta, \kappa_0) \geq 0$ for all $\beta \in \mathcal{B}([0,1])$ as expectation is a monotone operation. For $\beta \neq \beta_0$ let $A_\beta = \{\omega \in \Omega: \langle X_1(\omega),\beta-\beta_0 \rangle = 0 \}$. By the measurability of the map $X_1: \Omega \to \mathcal{L}^2([0,1])$, we have $A_{\beta} \in \mathcal{A}$ and, by (A2), $\mathbb{P}(A_\beta^c) >0$. Thus,
\begin{align*}
M(\beta, \kappa_0) & = \int_{\Omega} d_{\kappa_0}(f_{p_1(\beta_0)}, f_{p_1(\beta)}) d \mathbb{P} 
\\ & = \int_{A_\beta} d_{\kappa_0}(f_{p_1(\beta_0)}, f_{p_1(\beta)}) d \mathbb{P} +\int_{A_\beta^c} d_{\kappa_0}(f_{p_1(\beta_0)}, f_{\beta_1(g)}) d \mathbb{P}
\\ & = \int_{A_\beta^c} d_{\kappa_0}(f_{p_1(\beta_0)}, f_{p_1(\beta)}) d \mathbb{P}
\\ & >0,
\end{align*}
as, by (A2), $d_{\kappa_0}(f_{p_1(\beta_0)}, f_{p_1(\beta)})$ vanishes on $A_\beta$ and $d_{\kappa_0}(f_{p_1(\beta_0)}, f_{p_1(\beta)})>0$ on $A_\beta^c$. But, as $\mathbb{P}(A_\beta^c) >0$, integration on $A_\beta^c$ preserves the strict inequality. The proof is complete.

\end{proof}
For each fixed $\kappa_0>0$ and $\beta \in \mathcal{L}^2([0,1])$, let $M_{n}(\beta,\kappa_0)$ denote the empirical version of $M(\beta,\kappa_0)$, that is,
\begin{align}
\label{eq:A1}
M_{n}(\beta,\kappa_0) = \frac{1}{n} \sum_{i=1}^n \left[ \sum_{y \in \{0,1\}}  f^{1+\kappa_0}_{p_i(\beta)}(y)   - \left( 1 + \frac{1}{\kappa_0} \right)f_{p_i(\beta)}^{\kappa_0}(Y_i)  +\frac{1}{\kappa_0} \sum_{y \in \{0,1\}}  f_{p_i(\beta_0)}^{1+\kappa_0}(y) \right]
\end{align}
Notice that $M_{n}(\beta,\kappa_0)$ is the average of bounded i.i.d. random variables, hence, by the weak law of large numbers (WLLN), $M_{n}(\beta,\kappa_0) \xrightarrow{\mathbb{P}} M(\beta,\kappa_0)$. Lemma~\ref{lem:3} below provides a uniform law according to which $M_n(\beta, \kappa_0)-M(\beta, \kappa_0) \xrightarrow{\mathbb{P}} 0 $ uniformly over $\beta \in \Theta_{K_n}$, as $n \to \infty$.

\begin{lemma}
\label{lem:3}
Suppose that (A2) and (A4) are satisfied. Then, for every $0<\kappa_0<\infty$,
\begin{align*}
\sup_{\beta \in \Theta_{K_n}}|M_{n}(\beta,\kappa_0)-M(\beta, \kappa_0)| \xrightarrow{\mathbb{P}} 0.
\end{align*}
\end{lemma}
\begin{proof}

By the triangle inequality and the WLLN, it suffices to establish the uniform convergence for the classes of functions on $\mathcal{L}^2([0,1]) \times \{0,1\}$ given by
\begin{align*}
\mathcal{F}_n & = \left\{ f^{\kappa_0}_{H(\langle x, \beta \rangle)}(y), \beta \in \Theta_{K_n} \right\} \\
\mathcal{F}^{\prime}_n & = \left\{ \sum_{y \in \{0,1\}} f^{1+\kappa_0}_{H(\langle x, \beta \rangle)}(y), \beta \in \Theta_{K_n} \right\}.
\end{align*}
Here, as in Section~\ref{sec:2}, $f_{H(\langle x, \beta \rangle)}(y)$ denotes a Bernoulli density with probability of "success" $H(\langle x, \beta \rangle)$. 

Let $\mathbb{P}_n$ denote the empirical measure of $(X_i, Y_i), \ldots, (X_n, Y_n)$. We first show that $\mathcal{F}_n$ is a Glivenko-Cantelli class of functions \citep[p. 269]{VDV:1998}  so that 
\begin{align}
\label{eq:A2}
\sup_{\beta \in \Theta_{K_n}} \left| \int f^{\kappa_0}_{H(\langle x, \beta \rangle)}(y) d\mathbb{P}_n - \int f^{\kappa_0}_{H(\langle x, \beta \rangle)}(y) d\mathbb{P} \right| \xrightarrow{\mathbb{P}} 0,
\end{align}
To prove \eqref{eq:A2}, we develop a bound on the $\epsilon$-covering number of $\mathcal{F}_n$ in the $L_1(\mathbb{P}_n)$-norm, $N_1(\epsilon, \mathcal{F}_n, \mathbb{P}_n)$. Note that $N_1(\epsilon, \mathcal{F}_n, \mathbb{P}_n)$ is random, as it depends on the empirical measure $\mathbb{P}_n$. To bound this covering number, consider the set of real-valued functions on $\mathcal{L}^2([0,1])$ given by
\begin{align*}
\mathcal{G}_n = \left\{ \langle x, \beta \rangle, \beta \in \Theta_{K_n} \right\}.
\end{align*}
Clearly, $\mathcal{G}_n$ is at most a $K_n$-dimensional vector space of functions on $\mathcal{L}^2([0,1])$ and is therefore a Vapnik-Chervonenkis (VC) class of functions with VC index bounded above by $K_n+2$, see, e.g., \citep[pp. 275--276]{VDV:1998}. The class of functions $H(\mathcal{G}_n)$ is also VC with the same index, by the monotonicity of $H$ assumed in (A2) and Lemma 9.9 (viii) in \citep{Kos:2008}.  Furthermore, the class of functions $H(\mathcal{G}_n)$ has envelope equal to $1$, hence, by Lemma 19.15 in \citet{VDV:1998}, there exists a universal constant $C$ such that
\begin{align*}
\log(N_1(\epsilon, H(\mathcal{G}_n), \mathbb{P}_n)) \leq C (K_n+1) \log\left( \frac{1}{\epsilon} \right),
\end{align*}
for all small $\epsilon>0$. Observe next that the function $f_{p}(y)  =  p^{y}(1-p)^{1-y}, y \in \{0,1\}$, is Lipschitz with respect to $p$ with Lipschitz constant bounded above by $1$ so that the same bound holds for the set of functions on $\mathcal{L}^2([0,1]) \times \{0,1\}$ given by $\{f_{H(\langle x, \beta \rangle)}(y) , \beta \in \Theta_{K_n}\}$. To derive from this a bound on the covering number of $\mathcal{F}_n$ in the $L_1(\mathbb{P}_n)$-norm, simply note that the function $x \mapsto x^{\kappa_0}$ is increasing and bounded by $1$ for all $0 \leq x \leq 1$. Therefore, by Lemma 19.15 of \citet{VDV:1998} yet again,
\begin{align*}
\log(N_1(\epsilon, \mathcal{F}_n, \mathbb{P}_n)) \leq C (K_n+1) \log\left( \frac{1}{\epsilon} \right),
\end{align*}
for small $\epsilon>0$. By (A4), it now follows that, for such $\epsilon>0$, $\log(N_1(\epsilon, \mathcal{F})_n, \mathbb{P}_n))/n \xrightarrow{\mathbb{P}} 0$.  The uniform convergence in \eqref{eq:A2} now follows from \citet[Exercise 3.6]{van de Geer:2000} with $b_n = 2$ in that result, as the functions in $\mathcal{F}$ are uniformly bounded by $1$ for all $\beta \in \Theta_{K_n}$.

To complete the proof, we need to likewise establish that the set of functions $\mathcal{F}^{\prime}$ is Glivenko-Cantelli. For this we may consider the cases $y = 0$ and $y=1$ separately. The proof may be devised along similar lines: first bound the covering number of $H(\mathcal{G}_n)$ and then notice that the functions $x \mapsto x^{k_0}$ and $x \mapsto (1-x)^{k_0}$ are monotone and bounded for $0 \leq x \leq 1$. It follows that there exists a $C>0$ such that
\begin{align*}
\log(N_1(\epsilon, \mathcal{F}_n^{\prime }, \mathbb{P}_n)) \leq C (K_n+1) \log \left( \frac{1}{\epsilon} \right), 
\end{align*}
for all small $\epsilon>0$. Under (A4), the result follows again by an appeal to \citet[Exercise 3.6]{van de Geer:2000}.
\end{proof}

\begin{lemma}
\label{lem:4}
Suppose that assumptions (A1)--(A5) are satisfied. Then, for every $\kappa_0>0$,
\begin{align*}
M(\widehat{\beta}_n, \kappa_0) \xrightarrow{\mathbb{P}} M(\beta_0, \kappa_0).
\end{align*}

\end{lemma}

\begin{proof}
By Lemma~\ref{lem:1}, $\beta_0$ uniquely minimizes $M(\beta, \kappa_0)$ over all $\beta \in \mathcal{B}([0,1])$. Since $\widehat{\beta}_n \in \mathcal{B}([0,1]) \subset \mathcal{L}^2([0,1])$ for all $n \in \mathbbm{N}$, we clearly have
\begin{align}
\label{eq:A4}
0 < M(\widehat{\beta}_n, \kappa_0) - M(\beta_0, \kappa_0) = I + II + III,
\end{align}
where
\begin{align*}
I = M(\widehat{\beta}_n, \kappa_0) - M_n(\widehat{\beta}_n, \kappa_0),  \ II = M_n(\widehat{\beta}_n, \kappa_0) - M_n(\widehat{\beta}_n, \widehat{\kappa}_n), \ III = M_n(\widehat{\beta}_n, \widehat{\kappa}_n)-M(g_0, \kappa_0).
\end{align*}
Our proof consists of showing that each one of the terms in RHS of \eqref{eq:A4} converge to $0$ in probability. Starting with $I$, by definition of $\widehat{\beta}_n$ we have $\widehat{\beta}_n \in \Theta_{K_n}$, whence, by Lemma~\ref{lem:2},
\begin{align*}
\left|I\right| \leq \sup_{\beta \in \Theta_{K_n}}|M_{n}(\beta,\kappa_0)-M(\beta,\kappa_0)| \xrightarrow{\mathbb{P}} 0.
\end{align*}
To establish the convergence of $II$ to $0$, notice that, by (A1), we can assume that $\widehat{\kappa}_n \in (\kappa_0/2, 3\kappa_0/2)$ for all large $n$ with high probability. By the mean-value theorem, for every $x \in \mathcal{L}^2([0,1])$ and $y \in \{0,1\}$ there exists a  $\widetilde{\kappa}_n = \widetilde{\kappa}_n(x, y)$ such that $|\widetilde{\kappa}_n - \kappa_0| \leq |\widehat{\kappa}_n - \kappa_0| < \kappa_0/2$ and
\begin{align*}
f^{\widehat{\kappa}_n}_{H(\langle x, \widehat{\beta}_n \rangle)}(y) - f^{\kappa_0}_{H(\langle x, \widehat{\beta}_n \rangle)}(y) = (\widehat{\kappa}_n - \kappa_0) f^{\widetilde{\kappa}_n}_{H(\langle x, \widehat{\beta}_n \rangle) }(y)\log \left( f_{H(\langle x, \widehat{\beta}_n \rangle)}(y)\right).
\end{align*}
Now, 
\begin{align*}
\left|f^{\widetilde{\kappa}_n}_{H(\langle x, \widehat{\beta}_n \rangle) }(y)\log \left( f_{H(\langle x, \widehat{\beta}_n \rangle)}(y)\right) \right| & \leq \sup_{|\kappa-\kappa_0|<\kappa_0/2,\ p \in [0,1]}\left|f^{\kappa}_{p}(y)\log \left( f_{p}(y)\right) \right|
\\ & \leq \sup_{|\kappa-\kappa_0|<\kappa_0/2,\ p \in [0,1]} (1-p)^{\kappa} \left|\log(1-p)\right| + \sup_{|\kappa-\kappa_0|<\kappa_0/2,\ p \in [0,1]} p^{\kappa} \left|\log(p)\right|
\\ & \leq \sup_{\ p \in [0,1]} (1-p)^{\kappa_0/2} \left|\log(1-p)\right| +  \sup_{\ p \in [0,1]} p^{3\kappa_0/2} \left|\log(p)\right|
\\ & \leq M,
\end{align*}
for some finite $D = D(\kappa_0)$. Applying this together with the facts that $f_p(y) \leq 1$ and $\sum_{y \in\{0,1\}} f_p(y) = 1$ for every $p \in [0,1]$ yields
\begin{align*}
|II| & \leq \frac{1}{n}\sum_{i=1}^n \left[ \sum_{y \in \{0,1\}} \left| f_{H(\langle x, \widehat{\beta}_n \rangle)}^{\widehat{\kappa}_n}(y)- f_{H(\langle x, \widehat{\beta}_n \rangle)}^{\kappa_0}(y) \right| f_{H(\langle x, \widehat{\beta}_n \rangle)}(y) + \frac{|\widehat{\kappa}_n - \kappa_0|}{\kappa_0 \widehat{\kappa}_n} f_{H(\langle X_i, \widehat{\beta}_n \rangle)}^{\widehat{\kappa}_n}(Y_i)  \right]
\\ & \quad + \left(1+ \frac{1}{\widehat{\kappa}_n} \right)\frac{1}{n} \sum_{i=1}^n  \left| f_{H(\langle X_i, \widehat{\beta}_n \rangle)}^{\widehat{\kappa}_n}(Y_i)- f_{H(\langle X_i, \widehat{\beta}_n \rangle)}^{\kappa_0}(Y_i) \right|
\\ & \leq D  \left|\widehat{\kappa}_n - \kappa_0 \right| +   \frac{\left|\widehat{\kappa}_n - \kappa_0 \right|}{(\kappa_0-\delta)\kappa_0} + \left(1+ \frac{1}{\kappa_0-\delta}\right) D \left|\widehat{\kappa}_n - \kappa_0 \right|,
\end{align*}
for all large $n$ with high probability. Conclude that $II \xrightarrow{\mathbb{P}} 0 $.

To complete the proof we now show that $III \xrightarrow{\mathbb{P}} 0$. To see this, recall first that $\widehat{\beta}_n$ minimizes $M_n(\beta, \widehat{\kappa}_n) + \widehat{\lambda}_n \mathcal{J}(\beta)$ over $\beta \in \Theta_{K_n}$ and the abstract approximation to $\beta_0$ in $\Theta_{K_n}$, $\widetilde{\beta}_{K_n}$, satisfies $\widetilde{\beta}_{K_n} \in \Theta_{K_n} $. Hence, by the non-negativity of the map $ \beta \mapsto \widehat{\lambda}_n \mathcal{J}(\beta)$ and (A4), we obtain
\begin{align}
\label{eq:A5}
III &   \leq M_n(\widetilde{\beta}_{K_n}, \widehat{\kappa}_n) + \widehat{\lambda}_n \mathcal{J}(\widetilde{\beta}_{K_n}) - M(\beta_0, \kappa_0) \nonumber
\\ & = M_n(\widetilde{\beta}_{K_n}, \widehat{\kappa}_n) - M(\beta_0, \kappa_0) + o_\mathbb{P}(1) \nonumber
\\ & = \left\{M_n(\widetilde{\beta}_{K_n}, \widehat{\kappa}_n) - M_n(\beta_0, \kappa_0) \right\} + \left\{M_n(\beta_0, \kappa_0)  - M(\beta_0, \kappa_0) \right\} + o_\mathbb{P}(1).
\end{align}
By the WLLN, $M_n(\beta_0, \kappa_0) \xrightarrow{\mathbb{P}} M(\beta_0, \kappa_0)$. To show the convergence of the first term in the RHS of \eqref{eq:A5} to $0$ notice that, by the triangle inequality,
\begin{align}
\label{eq:A6}
\left|M_n(\widetilde{\beta}_{K_n}, \widehat{\kappa}_n) - M_n(\beta_0, \kappa_0)\right| \leq \left|M_n(\widetilde{\beta}_{K_n}, \widehat{\kappa}_n) - M_n(\widetilde{\beta}_{K_n}, \kappa_0)\right| + \left| M_n(\widetilde{\beta}_{K_n}, \kappa_0) -M_n(\beta_0, \kappa_0) \right|
\end{align}
Our previous argument employed in the treatment of $II$ immediately shows that the first term in the RHS of \eqref{eq:A6} is $o_\mathbb{P}(1)$, as $n \to \infty$. Furthermore, for every fixed $\kappa_0>0$, the mapping $[0,1] \to [0,1]: p \mapsto f_{p}^{\kappa_0}(y), y \in \{0,1\},$ is Lipschitz. Hence, there exists a finite $D^{\prime} = D^{\prime}(\kappa_0)$ such that
\begin{align*}
\left|f^{\kappa_0}_{H(\langle X_i,\widetilde{\beta}_{K_n} \rangle)}(y) - f^{\kappa_0}_{H(\langle X_i,\beta_0 \rangle)}(y) \right| \leq D^{\prime}(\kappa_0) \left|  H(\langle X_i,\widetilde{\beta}_{K_n} \rangle) - H(\langle X_i,\beta_0 \rangle) \right|.
\end{align*}
With the mean-value theorem, (A2) and (A5), we find
\begin{align*}
\left|f^{\kappa_0}_{H(\langle X_i,\widetilde{\beta}_{K_n} \rangle)}(y) - f^{\kappa_0}_{H(\langle X_i,\beta_0 \rangle)}(y) \right| &\leq  C D^{\prime}(\kappa_0)  \sup_{x \in \mathbbm{R}} H^{\prime}(x)   \left\|\widetilde{\beta}_{K_n} - \beta_0 \right\|,
\end{align*}
which leads to
\begin{align*}
\left|M_n(\widetilde{\beta}_{K_n}, \kappa_0) - M_n(\beta_0, \kappa_0) \right| \leq c_0 \left\|\widetilde{\beta}_{K_n} - \beta_0 \right\|,
\end{align*}
for some $c_0>0$. By (A4), $\| \widetilde{\beta}_{K_n} - \beta_0 \| \to 0$ as $K_n \to \infty$ so that $M_n(\widetilde{\beta}_{K_n}, \kappa_0) - M_n(\beta_0, \kappa_0) \xrightarrow{\mathbb{P}} 0$. We have thus shown that the RHS of \eqref{eq:A6} converges to zero in probability. Taking into account the convergence of $I$ and $II$ to $0$, we see that
\begin{align*}
o_\mathbb{P}(1)< III \leq o_\mathbb{P}(1).
\end{align*}
Consequently, $III = o_\mathbb{P}(1)$ and the proof is complete.
\end{proof}

\begin{proof}[Proof of Theorem 1]

The proof relies on Lemmas ~\ref{lem:2} and \ref{lem:4} as well as on  Lemma 14.3 and Theorem 2.12 of \citet{Kos:2008}. First, notice that, by Lemma~\ref{lem:2}, $M(\beta,\kappa_0)$ possesses a unique minimum at $\beta_0$ and by the continuity of the link function $H$, Lemma~\ref{lem:1}, (A5) and dominated convergence $M(\beta, \kappa_0)$ is continuous on $\mathcal{B}([0,1])$. By Lemma 14.3 of \citep{Kos:2008} it follows that for any sequence $\beta_k \in \mathcal{B}([0,1])$, $\lim_{k \to \infty} M(\beta_k, \kappa_0) = M(\beta_0,\kappa_0)$ implies $\|\beta_k-\beta_0\|_{\mathcal{B}} \to 0$. By exercise 14.6.2 of \citep{Kos:2008} there exists an increasing function cadlag (right continuous with left limits) function $f:[0,\infty] \to [0,\infty]$ that satisfies $f(0) = 0$, is continuous at $0$ and $\|\beta-\beta_0\|_{\mathcal{B}} \leq f( | M(\beta, \kappa_0) - M(\beta_0, \kappa_0)|)$  for all $\beta \in \mathcal{B}([0,1])$. Noting that our estimator $\widehat{\beta}_n \in \mathcal{B}([0,1])$ for all $n \in \mathbbm{N}$, we  find
\begin{align}
\label{eq:A7}
\left\|\widehat{\beta}_n - \beta_0 \right\|_{\mathcal{B}} \leq f\left( \left| M(\widehat{\beta}_n, \kappa_0) - M(\beta_0, \kappa_0)\right|\right).
\end{align}
By Lemma~\ref{lem:4}, $|M(\widehat{\beta}_n, \kappa_0)-  M(\beta_0, \kappa_0)| \xrightarrow{\mathbb{P}} 0$. Therefore, by the continuity of $f$ at $0$ and the fact that $f(0) = 0$, \eqref{eq:A7} entails that  $\|\widehat{\beta}_n - \beta_0\|_{\mathcal{B}} \xrightarrow{\mathbb{P}} 0$, as was to be proved. 
\end{proof}

We now introduce some useful notation that will be useful in the proof of Theorem~\ref{thm:2}. Let $\mathcal{G}$ denote a class of real-valued functions on $\mathcal{L}^2([0,1]) \times \mathbb{R}$. For $g \in \mathcal{G}$ we define
\begin{align*}
\|g\|_{\infty} = \sup_{x \in \mathcal{L}^2([0,1]), y \in \mathbb{R}}|g(x,y)|.
\end{align*}
The covering number in this uniform metric, $N_{\infty}(\epsilon, \mathcal{G})$, is defined as the smallest value of $N \in \mathbb{N}$ such that there exists a sequence $\{g_j\}_{j=1}^N$ with the property that
\begin{align*}
\sup_{g \in \mathcal{G}} \min_{j=1, \ldots, N} \|g-g_j\|_{\infty} \leq \epsilon. 
\end{align*}
The corresponding entropy, $H_{\infty}(\epsilon, \mathcal{G})$, is defined as the logarithm of the covering number, i.e.,  $H_{\infty}(\epsilon, \mathcal{G})= \log (N_{\infty}(\epsilon, \mathcal{G})) $. 

For a probability measure $\mathbb{P}$ we also define the bracketing number in the $\mathcal{L}^2(\mathbb{P})$-metric, $N_B(\epsilon, \mathcal{G}, \mathbb{P})$, as the smallest value of $N \in \mathbb{N}$ for which there exist $N$ pairs of functions $\{[g_j^L, g_j^U] \}$ such that $\|g_j^U-g_j^L\|_{\mathcal{L}^2(\mathbb{P})}\leq \epsilon $ for all $j=1, \ldots, N,$ and such that for each $g \in \mathcal{G}$, there is a $j=j(g) \in \{1, \ldots, N\}$ such that
\begin{align*}
g_j^L(x,y) \leq g(x,y) \leq g_j^U(x,y).
\end{align*}
The corresponding bracketing entropy, $H_B(\epsilon, \mathcal{G}, \mathbb{P})$, is defined as the logarithm of the bracketing number, i.e.,  $H_B(\epsilon, \mathcal{G}, \mathbb{P}) = \log  (N_B(\epsilon, \mathcal{G}, \mathbb{P}))$. 

For convenience throughout the rest of our proofs we use $c_0$ to denote generic positive constants whenever their value is not important for our argument. Thus, $c_0$ may change from appearance to appearance.

\begin{lemma}[Bracketing entropy]
\label{lem:5}
For $\beta \in \Theta_{K_n}$ and $\kappa \in (\kappa_0-\delta_2, \kappa_0+\delta_2)$ for $\delta_2>0$ such that $\kappa_0-\delta_2>0$ define the real-valued function $z: \mathcal{L}^2([0,1]) \times \{0,1\} \to \mathbbm{R}$,
\begin{align*}
z_{\beta,\kappa}(x,y) = \sum_{y \in \{0,1\}} \left\{f_{H(\langle x,\widetilde{\beta}_{K_n} \rangle )}^{1+\alpha}(y) - f_{H(\langle x,\beta \rangle )}^{1+\kappa}(y)\right\} - \left(1+\frac{1}{\kappa}  \right) \left\{f_{H(\langle x,\widetilde{\beta}_{K_n} \rangle )}^{\kappa}(y) - f_{H(\langle x,\beta \rangle )}^{\kappa}(y) \right\} ,
\end{align*} 
and the class of functions
\begin{align*}
\mathcal{Z}_{n, \delta_1, \delta_2} = \left\{z_{\beta,\kappa}(x,y), \beta \in \Theta_{K_n}, \|\beta-\widetilde{\beta}_{K_n} \| \leq \delta_1, |\kappa-\kappa_0| <\delta_2  \right\}.
\end{align*}
Let $\mathbb{P}_{X,Y}$ denote the probability measure (distribution) induced by the random variables $(X,Y) \in \mathcal{L}^2([0,1]) \times \{0,1\}$. Then, for each sufficiently small $\delta_1>0$ and $\delta_2>0$, there exists a positive constant $A = A(\delta_1, \delta_2)$ such that	for all $n \geq n_0$ we have
\begin{align*}
H_{B}\left(\epsilon, \mathcal{Z}_{n, \delta_1, \delta_2}, \mathbb{P}_{X,Y}\right) \leq 3 K_n \log \left(1 + \frac{A}{\epsilon} \right), \quad \epsilon>0.
\end{align*}
\end{lemma}

\begin{proof}

By Lemma 2.1 of \citet{van de Geer:2000}, we have
\begin{align*}
H_{B}\left(\epsilon, \mathcal{Z}_{n, \delta_1, \delta_2}, \mathbb{P}_{X,Y}\right) \leq H_{\infty}\left(\epsilon/2, \mathcal{Z}_{n, \delta_1, \delta_2} \right), \quad \epsilon>0.
\end{align*}
Hence it suffices to bound the uniform entropy. By (A5), we may assume without loss of generality that $\|x\| \leq C$. For any two functions in $\mathcal{Z}_{n,\delta_1,\delta_2}$, $z_{\beta_1, \kappa_1}(x,y)$ and $z_{\beta_2, \kappa_2}(x,y)$, we have
\begin{align*}
\left|z_{\beta_1, \kappa_1}(x,y) - z_{\beta_2, \kappa_2}(x,y) \right| & \leq \left|z_{\beta_1, \kappa_1}(x,y) - z_{\beta_1, \kappa_2}(x,y) \right| + \left|z_{\beta_1, \kappa_2}(x,y) - z_{\beta_2, \kappa_2}(x,y) \right|
\\ & \leq \sum_{y \in \{0,1\}} \left\{ \left|f_{H(\langle x,\beta_1 \rangle )}^{1+\kappa_1}(y) - f_{H(\langle x,\beta_1 \rangle )}^{1+\kappa_2}(y)  \right| +   \left|f_{H(\langle x,\widetilde{\beta}_{K_n} \rangle )}^{1+\kappa_1}(y) - f_{H(\langle x,\widetilde{\beta}_{K_n} \rangle )}^{1+\kappa_2}(y) \right| \right\}
\\ & \quad + \frac{\left|\kappa_1-\kappa_2\right|}{|\kappa_0-\delta|^2} + \left|f_{H(\langle x,\beta_1 \rangle )}^{\kappa_1}(y) - f_{H(\langle x,\beta_1 \rangle )}^{\kappa_2}(y)  \right| +   \left|f_{H(\langle x,\widetilde{\beta}_{K_n} \rangle )}^{\kappa_1}(y) - f_{H(\langle x,\widetilde{\beta}_{K_n} \rangle )}^{\kappa_2}(y) \right|
\\ & \quad + \sum_{y \in \{0,1\}}  \left|f_{H(\langle x,\beta_1 \rangle )}^{1+\kappa_2}(y) - f_{H(\langle x,\beta_2 \rangle )}^{1+\kappa_2}(y)  \right| + \frac{\kappa_2 +1}{\kappa_0-\delta} \left|f_{H(\langle x,\beta_1 \rangle )}^{\kappa_2}(y) - f_{H(\langle x,\beta_2 \rangle )}^{\kappa_2}(y) \right|
\end{align*}
For all sufficiently small $\delta_1>0$, $\|\beta - \widetilde{\beta}_{K_n}\| <\delta_1$ implies that $H(\langle x, \beta \rangle)$ is bounded away from $0$ for all large $n$. To see this, first observe that, by (A2), $H(\langle x, \beta_0 \rangle)$ is bounded away from $0$ and $1$, as $|\langle x, \beta_0 \rangle| \leq C \|\beta_0 \|$ and $H$ is bijective. Additionally,
\begin{align*}
\left|H(\langle x, \beta \rangle) - H(\langle x, \beta_0 \rangle)\right| & \leq \left|H(\langle x, \beta \rangle) - H(\langle x, \widetilde{\beta}_{K_n} \rangle)\right| + \left|H(\langle x, \widetilde{\beta}_{K_n} \rangle) - H(\langle x, \beta_0 \rangle)\right|
\\ & \leq \sup_{x \in \mathbbm{R}} H^{\prime}(x) C \delta_1 + o(1),
\end{align*}
as $n \to \infty$. Now, for $p$ bounded away from $0$ and $1$, say $0<p_{\min} \leq p \leq p_{\max}<1$, it is easy to see that the function $p \mapsto f_{p}^{\kappa}(y)$ is Lipschitz for all $\kappa \in (\kappa_0-\delta_2, \kappa_0+\delta_2)$ with Lipschitz constant equal to $(k_0+\delta_2)(1+p_{\max})/(p_{\min}(1-p_{\max}))$. Therefore, we further find
\begin{align*}
\left|z_{\beta_1, \kappa_1}(x,y) - z_{\beta_2, \kappa_2}(x,y) \right|  & \leq 4 \sup_{y \in \{0,1\}, p \in [0,1], |\kappa - \kappa_0|<\delta_2} \left|f^{1+\kappa}_{p}(y) \log(f_{p}(y)) \right| \left|\kappa_1-\kappa_2\right|+  \frac{\left|\kappa_1-\kappa_2\right|}{|\kappa_0-\delta_2|^2} 
\\ & \quad +  2 \sup_{y \in \{0,1\}, p \in [0,1], |\kappa - \kappa_0|<\delta_2} \left|f^{\kappa}_{p}(y) \log(f_{p}(y)) \right| \left|\kappa_1-\kappa_2\right| 
\\ & \quad + \sup_{x \in \mathbbm{R}} H^{\prime}(x) C \frac{(1+k_0+\delta_2)(1+p_{\max})}{p_{\min}(1-p_{\max})} \left\|\beta_1 - \beta_2 \right\|
\\& \leq c_0 \left|\kappa_1-\kappa_2 \right| + c_0 \left\|\beta_1-\beta_{\widetilde{K}_n} \right\| + c_0 \left\|\beta_2-\beta_{\widetilde{K}_n} \right\|
\end{align*}
for some global $c_0>0$ depending only on $\delta_2$.  Conclude  that 
\begin{align*}
H_{\infty}\left(\epsilon, \mathcal{Z}_{n, \delta_1, \delta_2} \right) \leq H\left(\epsilon/(2c_0), \{|\kappa-\kappa_0|\leq \delta_2 \} \right) + 2H\left( \epsilon/(4 c_0), \{\beta \in \Theta_{K_n}, \|\beta-\widetilde{\beta}_{K_n}\| \leq \delta_1 \} \right).
\end{align*}
These two entropy numbers may be bounded by Lemma 2.5 and Corollary 2.6 of \citet{van de Geer:2000} respectively leading to
\begin{align*}
H_{\infty}\left(\epsilon, \mathcal{Z}_{n, \delta_1, \delta_2} \right) \leq  \log\left( \frac{8 c_0\delta_2}{\epsilon}+1 \right) + 2K_n \log \left( \frac{16c_0 \delta_1}{\epsilon}+1 \right) \leq 3 K_n \log\left(\frac{A}{\epsilon}+1 \right),
\end{align*}
for $A = 8 c_0\delta_2 \vee 16 c_0\delta_1$. The proof is complete.
\end{proof}

\begin{proof}[Proof of Theorem~\ref{thm:2}]

Observe that, since $\widehat{\beta}_n$ minimizes $M_n(\beta, \widehat{\kappa}_n) + \widehat{\lambda}_n \mathcal{J}(\beta)$ over $\beta \in \Theta_{K_n}$ and $\widetilde{g}_{K_n} \in \Theta_{K_n}$, we have
\begin{align}
\label{eq:A8}
M\left(\widehat{\beta}_n, \widehat{\kappa}_n\right)-M\left(\widetilde{\beta}_{K_n}, \widehat{\kappa}_n\right) \leq U_n\left(\widetilde{\beta}_{K_n}, \widehat{\beta}_n, \widehat{\kappa}_n \right) +  \widehat{\lambda}_n \mathcal{J}(\widetilde{\beta}_{K_n}),
\end{align}
where we have dropped the non-negative term $\widehat{\lambda}_n \mathcal{J}(\widehat{g}_n)$ and rearranged terms. Here, $U_n(f, g, \alpha): \Theta_{K_n} \times \Theta_{K_n} \times \mathbbm{R}$ denotes the centered process given by
\begin{align*}
U_n(f, g, \alpha) =  \{M_n(f, \alpha) - M(f,\alpha\} - \{M_n(g, \alpha) - M(g, \alpha) \}.
\end{align*}
Our proof consists of two steps. In the first step we establish a quadratic lower bound of the LHS of \eqref{eq:A8} in terms of $\pi(\widehat{\beta}_n, \widetilde{\beta}_{K_n})$. In particular, we will show that there exist constants $\eta, L>0$ such that
\begin{align}
\label{eq:A9}
M\left(\widehat{\beta}_n, \widehat{\kappa}_n\right)-M\left(\widetilde{\beta}_{K_n}, \widehat{\kappa}_n\right) \geq \eta \left|\pi(\widehat{\beta}_n, \widetilde{\beta}_{K_n}) \right|^2 - L \left\| \widetilde{\beta}_{K_n} - \beta_0 \right\|  \pi(\widehat{\beta}_n, \widetilde{\beta}_{K_n}),
\end{align}
for all large $n$ with high probability. In the second step we derive the modulus continuity of $U_n\left(\widetilde{\beta}_{K_n}, \widehat{\beta}_n, \widehat{\kappa}_n \right)$ in terms of $\pi(\widehat{\beta}_n, \widetilde{\beta}_{K_n})$ and show that
\begin{align}
\label{eq:A10}
U_n\left(\widetilde{\beta}_{K_n}, \widehat{\beta}_n, \widehat{\kappa}_n \right) = O_{\mathbb{P}}(1) \left\{ \gamma_n \pi(\widehat{\beta}_n,  \widetilde{\beta}_{K_n}) \vee \gamma_n^2 \right\},
\end{align}
where $\gamma_n^2 = K_n \log n/n$. In combination, \eqref{eq:A9} and \eqref{eq:A10} imply the result of the theorem. To see this, plug \eqref{eq:A9} and \eqref{eq:A10} into \eqref{eq:A8} and rearrange to get
\begin{align}
\label{eq:A11}
\left|\pi(\widehat{g}_n, \widetilde{g}_{K_n}) \right|^2 \leq L^{\prime} \left( \gamma_n \pi(\widehat{\beta}_n, \widetilde{\beta}_{K_n}) \vee \gamma_n^2 \right)  + L^{\prime} \left\|  \widetilde{\beta}_{K_n} - \beta_0 \right\|  \pi(\widehat{\beta}_n, \widetilde{\beta}_{K_n}) + L^{\prime}\widehat{\lambda}_n \mathcal{J}\left(\widetilde{\beta}_{K_n} \right),
\end{align}
for some sufficiently large $L^{\prime}$ and all large $n$ with high probability. Notice then that we only need to prove the theorem for $\gamma_n  \pi(\widehat{g}_n, \widetilde{g}_{K_n}) \geq \gamma_n^2$, or, equivalently, $\pi(\widehat{g}_n, \widetilde{g}_{K_n}) \geq \gamma_n$,  otherwise the result clearly holds. With this simplification, \eqref{eq:A11} is an inequality of the form $x_0^2 \leq b x_0 + c$ for $x_0 = \pi(\widehat{\beta}_n, \widetilde{\beta}_{K_n})  \geq 0$, $b = L^{\prime}( \left\|  \widetilde{\beta}_{K_n}-\beta_0 \right\| + \gamma_n)$ and $c = L^{\prime}\mathcal{J}\left(\widetilde{\beta}_{K_n} \right)$. This means that $x_0$ must be less than or equal to the positive root of $x^2-bx-c=0$, that is,
\begin{align*}
0 \leq x_0 \leq \frac{b+\sqrt{b^2+4c}}{2} \leq b + \sqrt{c},
\end{align*}
and after substituting for $x_0$, $b$ and $c$,
\begin{align*}
\pi(\widehat{\beta}_n, \widetilde{\beta}_{K_n}) \leq L^{\prime} \left( \left\| \widetilde{\beta}_{K_n} - \beta_0 \right\| + \gamma_n\right) + \sqrt{L^{\prime} \widehat{\lambda}_n \mathcal{J}\left(\widetilde{\beta}_{K_n} \right)}.
\end{align*}
Squaring and using the inequality $(x+y)^2 \leq 2x^2 + 2y^2$ twice, we get
\begin{align*}
\left| \pi(\widehat{\beta}_n, \widetilde{\beta}_{K_n}) \right|^2 \leq O_{\mathbb{P}}(1)\gamma_n^2 + O_{\mathbb{P}}(1)\left\| \widetilde{\beta}_{K_n} -\beta_0\right\|^2 + O_{\mathbb{P}}(1)  \widehat{\lambda}_n \mathcal{J}\left(\widetilde{\beta}_{K_n}\right).
\end{align*}
To pass from $\pi(\widehat{\beta}_n, \widetilde{\beta}_{K_n})$ to $\pi(\widehat{\beta}_n, g_0)$ simply note that
\begin{align*}
\left|\pi(\widehat{\beta}_n, \beta_0)\right|^2 &\leq 2 \left|\pi(\widehat{\beta}_n, \widetilde{\beta}_{K_n}) \right|^2 + 2C^2 \left\|\widetilde{\beta}_{K_n}-\beta_0 \right\|^2 
\\ &= O_{\mathbb{P}}(1)\gamma_n^2 + O_{\mathbb{P}}(1)\left\| \widetilde{\beta}_{K_n} - \beta_0 \right\|^2 + O_{\mathbb{P}}(1)  \widehat{\lambda}_n \mathcal{J}\left(\widetilde{\beta}_{K_n}\right),
\end{align*}
as we wanted to show.

To prove \eqref{eq:A9} and \eqref{eq:A10} let us first observe that by Theorem~\ref{thm:1} and (A1) we may restrict attention to the set $F_{\delta} \subset \mathcal{B}([0,1]) \times \mathbbm{R}$ given by
\begin{align}
\label{eq:A12}
F_{\delta} = \left\{ \left\|\beta - \beta_0 \right\|_{\infty} < \delta, \left|\kappa-\kappa_0\right|<\delta \right\},
\end{align}
for some sufficiently small $\delta>0$. To establish \eqref{eq:A9} it thus suffices to prove
\begin{align}
\label{eq:A13}
M\left(\beta, \kappa \right)-M\left(\widetilde{\beta}_{K_n}, \kappa\right) \geq \eta \left|\pi(\beta, \widetilde{\beta}_{K_n}) \right|^2 - L \left\| \beta - \widetilde{\beta}_{K_n} \right\|  \pi(\beta, \widetilde{\beta}_{K_n})
\end{align}
for some $\eta, L>0$, uniformly in $(\beta, \kappa) \in F_{\delta}$.  To prove \eqref{eq:A13}, first observe that
\begin{align}
\label{eq:A14}
M\left(\beta, \kappa \right)-M\left(\widetilde{\beta}_{K_n}, \kappa\right) & = \sum_{y \in \{0,1\}} \mathbb{E}\left\{f_{p_1(\beta)}^{1+\kappa}(y) - f_{p_1(\widetilde{\beta}_{K_n})}^{1+\kappa}(y) \right\} \nonumber
\\ & \quad  - \left(1 + \frac{1}{\kappa} \right) \sum_{y \in \{0,1\}} \mathbbm{E}\left\{ f_{p_1(\beta_0)}(y) \left[f_{p_1(\beta)}^{\kappa}(y) -f_{p_1(\widetilde{\beta}_{K_n})}^{\kappa}(y)  \right] \right\}.
\end{align} 
Consider next a first order Taylor expansion with Lagrange remainder of $f_{p_1(\beta)}^{1+\kappa}(y)$ and $f_{p_1(\beta)}^{\kappa}(y)$ about $f_{p_1(\widetilde{\beta}_{K_n})}^{1+\kappa}(y)$ and $f_{p_1(\widetilde{\beta}_{K_n})}^{\kappa}(y)$, respectively. With the formulas
\begin{align*}
\frac{\partial f_{p}^{\kappa}(y)}{\partial p} &= \kappa f_{p}^{\kappa}(y)\left\{ \frac{y-p}{p(1-p)} \right\} \\ 
\frac{\partial^2 f_{p}^{\kappa}(y)}{\partial p^2} & = \kappa \frac{f_{p}^{\kappa}(y)}{p^2(1-p)^2} \left\{(\kappa-1)p^2 + 2(1-\kappa)yp + \kappa y^2 -y \right\},
\end{align*}
which are valid for every $\kappa>0$, we deduce the existence of random variables $p_{\star} = p_{\star}(X_1,y, \kappa, \beta)$ and $p_{\star \star} = p_{\star \star}(X_1,Y_1, \kappa, \beta)$ such that $\max\{|p_{\star}-p_1(\widetilde{\beta}_{K_n})|,|p_{\star\star}-p_1(\widetilde{\beta}_{K_n})|\} \leq |p_1(\beta)-p_1(\widetilde{\beta}_{K_n})|$ and
\begin{align*}
f_{p_1(\beta)}^{1+\kappa}(y) - f_{p_1(\widetilde{\beta}_{K_n})}^{1+\kappa}(y) & = (1+\kappa)  f_{p_1(\widetilde{g}_{K_n})}^{1+\kappa}(y) \left\{ \frac{y-p_1(\widetilde{\beta}_{K_n})}{p_1(\widetilde{\beta}_{K_n})(1-p_1(\widetilde{\beta}_{K_n}))} \right\}(p_1(\beta) - p_1(\widetilde{\beta}_{K_n}))
\\ & \quad + (1+\kappa) \frac{f_{p_{\star}}^{1+\kappa}(y)}{2 p_{\star}^2(1-p_{\star})^2}\left\{\alpha p_{\star}^2 - 2\alpha yp_{\star}+ (\kappa+1)y^2-y \right\} \left|p_1(\beta) - p_1(\widetilde{\beta}_{K_n})\right|^2 
\\ & = (1+\kappa)  f_{p_1(\beta_0)}^{1+\kappa}(y) \left\{ \frac{y-p_1(\beta_0)}{p_1(\beta_{0})(1-p_1(\beta_{0}))} \right\}(p_1(\beta) - p_1(\widetilde{\beta}_{K_n}))
\\ & \quad + (1+\kappa) \frac{f_{p_{\star}}^{1+\kappa}(y)}{2 p_{\star}^2(1-p_{\star})^2}\left\{\kappa p_{\star}^2 - 2\kappa yp_{\star}+ (\kappa+1)y^2-y \right\} \left|p_1(\beta) - p_1(\widetilde{\beta}_{K_n})\right|^2 
\\ & \quad +(1+\kappa) \left[ f_{p_1(\widetilde{\beta}_{K_n})}^{1+\kappa}(y) \left\{ \frac{y-p_1(\widetilde{\beta}_{K_n})}{p_1(\widetilde{\beta}_{K_n})(1-p_1(\widetilde{\beta}_{K_n}))} \right\} -    f_{p_1(\beta_0)}^{1+\kappa}(y) \left\{  \frac{y-p_1(\beta_0)}{p_1(\beta_{0})(1-p_1(\beta_{0}))} \right\}   \right]
  \\ & \quad  \times (p_1(\beta) - p_1(\widetilde{\beta}_{K_n}))
\end{align*}
as well as
\begin{align*}
f_{p_1(\beta)}^{\kappa}(y) -f_{p_1(\widetilde{\beta}_{K_n})}^{\kappa}(y)  & =\kappa f_{p_1(\widetilde{\beta}_{K_n})}^{\kappa}(y) \left\{ \frac{y-p_1(\widetilde{\beta}_{K_n})}{p_1(\beta_{K_n})(1-p_1(\beta_{K_n}))} \right\}  \left(p_1(\beta)-p_1(\widetilde{\beta}_{K_n})) \right) \\
  & \quad + \kappa \frac{f_{p_{\star\star}}^{\kappa}(y)}{2 p_{\star \star}^2(1-p_{\star  \star})^2}\left\{(\kappa-1)p_{\star\star}^2 + 2(1-\kappa) y p_{\star \star}+ \kappa y^2-y \right\}  \left|p_1(\beta) - p_1(\widetilde{\beta}_{K_n})\right|^2   
 \\ & =  \kappa f_{p_1(\beta_0)}^{\kappa}(y) \left\{ \frac{y-p_1(\beta_0)}{p_1(\beta_{0})(1-p_1(\beta_{0}))} \right\}(p_1(\beta) - p_1(\widetilde{\beta}_{K_n}))
 \\ & \quad + \kappa \frac{f_{p_{\star\star}}^{\kappa}(y)}{2 p_{\star \star}^2(1-p_{\star  \star})^2}\left\{(\kappa-1)p_{\star\star}^2 + 2(1-\kappa) y p_{\star \star}+ \kappa y^2-y \right\}  \left|p_1(\beta) - p_1(\widetilde{\beta}_{K_n})\right|^2  
\\ & \quad +\kappa \left[ f_{p_1(\widetilde{\beta}_{K_n})}^{\kappa}(y) \left\{ \frac{Y_1-p_1(\widetilde{\beta}_{K_n})}{p_1(\widetilde{\beta}_{K_n})(1-p_1(\widetilde{\beta}_{K_n}))} \right\} -    f_{p_1(\beta_0)}^{\kappa}(y) \left\{  \frac{y-p_1(\beta_0)}{p_1(\beta_{0})(1-p_1(\beta_{0}))} \right\}   \right]
  \\ & \quad  \times \left\{p_1(\beta) - p_1(\widetilde{\beta}_{K_n})\right\}.
\end{align*}
Plugging these expansions into \eqref{eq:A14} and simplifying we get
\begin{align*}
M\left(\beta, \kappa \right)-M\left(\widetilde{\beta}_{K_n}, \kappa \right) & = (1+\kappa	) \sum_{y \in \{0,1\}} \mathbb{E} \left\{ \left(u_n(X_1, y, \widetilde{\beta}_{K_n}, \kappa) - v_n(X_1,y, \widetilde{\beta}_{K_n},  \kappa) \right)  \left(p_1(\beta)-p_1(\widetilde{\beta}_{K_n}) \right) \right\} \nonumber
\\ & \quad + (1+\kappa)  \sum_{y \in \{0,1\}} \mathbb{E} \left\{ \left(u_n^{\prime}(X_1, y, \beta, \kappa) - v_n^{\prime}(X_1, y, \beta, \kappa) \right) \left|p_1(\beta)-p_1(\widetilde{\beta}_{K_n}) \right|^2 \right\},
\end{align*}
where for convenience we have defined
\begin{align*}
u_n(X_1, y, \widetilde{\beta}_{K_n}, \alpha) & = f_{p_1(\widetilde{\beta}_{K_n})}^{1+\kappa}(y) \left\{ \frac{y-p_1(\widetilde{\beta}_{K_n})}{p_1(\widetilde{\beta}_{K_n})(1-p_1(\widetilde{\beta}_{K_n}))} \right\} -    f_{p_1(\beta_0)}^{1+\kappa}(y) \left\{  \frac{y-p_1(\beta_0)}{p_1(\beta_{0})(1-p_1(\beta_{0}))} \right\}
\\ v_n(X_1, y, \widetilde{\beta}_{K_n}, \kappa) & = f_{p_1(\widetilde{\beta}_{K_n})}^{\kappa}(y) f_{p_1(\beta_0)}(y) \left\{ \frac{y-p_1(\widetilde{\beta}_{K_n})}{p_1(\widetilde{\beta}_{K_n})(1-p_1(\widetilde{\beta}_{K_n}))} \right\} -    f_{p_1(\beta_0)}^{1+\kappa}(y) \left\{  \frac{y-p_1(\beta_0)}{p_1(\beta_{0})(1-p_1(\beta_{0}))} \right\}
\\ u_n^{\prime}(X_1, y, \beta, \kappa) & = \frac{f_{p_{\star}}^{1+\kappa}(y)}{2 p_{\star}^2(1-p_{\star})^2}\left\{\kappa p_{\star}^2 - 2\kappa yp_{\star}+ (\kappa+1)y^2-y \right\} 
\\ v_n^{\prime}(X_1, y, \beta, \kappa) & = \frac{f_{p_{\star\star}}^{\kappa}(y)f_{p_1(\beta_0)}(y)}{2 p_{\star \star}^2(1-p_{\star  \star})^2}\left\{(\kappa-1)p_{\star\star}^2 + 2(1-\kappa) yp_{\star \star}+ \kappa y^2-y \right\},
\end{align*}
where $u_n^{\prime}(X_1, y, \beta, \kappa)$ and $v_{n}^{\prime}(X_1, y, \beta, \kappa)$ depend on $\beta$ through $p_{\star}$ and $p_{\star \star}$ respectively.

We first show that for sufficiently small $\delta$ in the definition of $F_{\delta}$, there exists an $\eta>0$
\begin{align}
\label{eq:A15}
(1+\kappa)  \sum_{y \in \{0,1\}} \mathbb{E} \left\{ \left(u_n^{\prime}(X_1, y, \beta, \kappa) - v_n^{\prime}(X_1, y, \beta, \kappa) \right) \left|p_1(\beta)-p_1(\widetilde{\beta}_{K_n}) \right|^2 \right\} \geq \eta \left|\pi(\beta, \widetilde{\beta}_{K_n}) \right|^2,
\end{align}
for all sufficiently large $n$, uniformly in $(\beta, \kappa) \in F$. To see this, observe first that, by definition of $p_{\star}$ and our  assumptions, we find
\begin{align*}
\left|p_{\star}-p_1(\beta_0)\right| & \leq \left|p_{\star}-p_1(\widetilde{\beta}_{K_n})\right| + \left|p_1(\widetilde{\beta}_{K_n})-p_1(\beta_0)\right| 
\\ & \leq c_0\left\|\beta-\beta_0\right\|_{\infty} + c_0 \left\|\widetilde{\beta}_{K_n} - \beta_0\right\|
\\ &  \leq c_0 \delta + o(1)
\end{align*}
A similar argument also shows that $|p_{\star \star}-p_1(g_0)| < c_0 \delta+o(1)$. Recall next that $p_1(\beta_0)$ is bounded away from $0$ and $1$ with probability one (see the proof of Lemma~\ref{lem:5}). Hence, $p_{\star}$ and $p_{\star \star}$ are also bounded away from $0$ and $1$ for all $\delta \leq \delta_{0}$ and $n \geq n_0$ with probability one. By continuity, there exists a $\delta_{0}^{\prime} \leq \delta_0$ such that for all $\delta \leq \delta_{0}^{\prime}$ we have
\begin{align*}
u_n^{\prime}(X_1,y,\beta,\kappa) - v_n^{\prime}(X_1,y,\beta, \kappa) & \geq \frac{f_{p_1(\beta_0)}^{1+\kappa}(y)}{8p_1^2(\beta_0)(1-p_1(\beta_0))^2}\left\{p_1^2(\beta_0)-2yp_1(\beta_0) + y^2 \right\} \nonumber
\\ &   = \frac{f_{p_1(\beta_0)}^{1+\kappa}(y)}{8p_1^2(\beta_0)(1-p_1(\beta_0))^2}\left|y-p_1(\beta_0)\right|^2,
\end{align*}
for all $n \geq n_0$ with probability one. Summing over $y \in \{0,1\}$ we see that there exists a $c_0>0$ such that
\begin{align*}
\sum_{y \in \{0,1\}} \left\{u_n^{\prime}(X_1,y,\beta,\kappa) - v_n^{\prime}(X_1,y,\beta, \kappa) \right\} \geq c_0
\end{align*}
for all $n \geq n_0$ with probability one. Therefore,
\begin{align*}
\sum_{y \in \{0,1\}} \mathbb{E} \left\{ \left(u_n^{\prime}(X_1, y, \beta, \kappa) - v_n^{\prime}(X_1, y, \beta, \kappa) \right) \left|p_1(\beta)-p_1(\beta_{K_n}) \right|^2 \right\} \geq c_0 \mathbb{E} \left\{\left|p_1(\beta)-p_1(\widetilde{\beta}_{K_n}) \right|^2 \right\}.
\end{align*}
Now, by the mean-value theorem, there exists a (random) $x^{\star} \in \mathbbm{R}$ such that $|x^{\star} - \langle X_1, \widetilde{\beta}_{K_n} \rangle| \leq |\langle X_1, \beta \rangle - \langle X_1, \widetilde{\beta}_{K_n} \rangle|$ and
\begin{align*}
\left|p_1(\beta) -p_1(\widetilde{\beta}_{K_n})\right|^2 = \left|H\left(\langle X_1, \beta \rangle \right) -  H\left(\langle X_1, \widetilde{\beta}_{K_n} \rangle \right) \right|^2 = \left|H^{\prime}(x^{\star}) \right|^2 \left|\langle X_1, \beta - \widetilde{\beta}_{K_n} \rangle\right|^2.
\end{align*}
For $\beta \in \{f \in \mathcal{B}([0,1]): \|\beta-\beta_0\|_{\infty} <\delta \}$, by (A5) and (A4), we find
\begin{align*}
\left|x^{\star} - \langle X_1, \beta_0 \rangle\right| \leq c_0 \delta +o(1),
\end{align*}
for some $c_0>0$ with probability one. Since  $|\langle X_1, \beta_0 \rangle|$ is bounded away from zero and one with probability one and $H$ is strictly increasing, this entails the existence of a $\delta_0^{\prime \prime}$ such that $H^{\prime}(x^{\star}) \geq c_0 >0$ with probability one for all $ \delta  \leq \delta_0^{\prime \prime} \leq \delta_0^{\prime}$  and $n \geq n_0^{\prime}$, whence we obtain
\begin{align*}
\mathbb{E}\left\{ \left| p_1(\beta)-p_1(\widetilde{\beta}_{K_n})\right|^2 \right\} \geq c_0^2 \mathbb{E}\left\{ \left| \langle X_1, \beta - \widetilde{\beta}_{K_n} \rangle\right|^2 \right\} = c_0^2 |\pi(\beta,\widetilde{\beta}_{K_n})|^2,
\end{align*}
for all $\delta \leq \delta_0^{\prime \prime}$ and $n \geq n_0^{\prime}$ with probability one, 
for all large $n$. Combining the above lower bounds, we find
\begin{align*}
(1+\kappa)\sum_{y \in \{0,1\}} \mathbb{E} \left\{ \left(u_n^{\prime}(X_1, y, \beta, \kappa) - v_n^{\prime}(X_1, y, \beta, \kappa) \right) \left|p_1(\beta)-p_1(\beta_{K_n}) \right|^2 \right\} \geq \eta \left|\pi(\beta,\widetilde{\beta}_{K_n})\right|^2,
\end{align*}
for all $n \geq n_0\vee n_0^{\prime}$ and $\delta \leq \delta_0^{\prime \prime}$. Here, $\eta = c_0 (1+\kappa_0-\delta)>0$ and $(1+\kappa_0-\delta)$  is strictly positive for all $\delta \leq (1+\kappa_0)/2$, as $c_0>0$. Choosing $\delta \leq  \delta_0^{\prime \prime} \wedge (\kappa_0+1)/2$ ensures the positivity of $\eta>0$, completing the proof of \eqref{eq:A15}.

We next show that 
\begin{align}
\label{eq:A16}
\sum_{y \in \{0,1\}}&\left| \mathbb{E}  \left\{ \left(u_n(X_1, y, \widetilde{\beta}_{K_n}, \kappa) - v_n(X_1,y, \widetilde{\beta}_{K_n},  \kappa) \right)  \left(p_1(\beta)-p_1(\widetilde{\beta}_{K_n}) \right) \right\} \right| \nonumber \\ & \quad \leq L \left\|\widetilde{\beta}_{K_n} - \beta_0 \right\| \pi(\beta, \widetilde{\beta}_{K_n}),
\end{align}
for some finite $L>0$, uniformly in $(\beta, \kappa) \in F_{\delta}$. To prove this, notice first that by the boundedness of $p_1(\beta_0)$ away from $0$ and $1$ and the fact that $|p_1(\widetilde{\beta}_{K_n}) - p_1(\beta_0)| \to 0$ as $n \to \infty$ with probability one, it may also be assumed that $p_1(\widetilde{\beta}_{K_n})$ is also bounded away from $0$ and $1$ for all large $n$. For $p$ bounded away from $0$ and $1$ the functions $p \mapsto f_p^{\kappa}(y)$ and  $p \mapsto (y-p)/p(1-p)$ are for all $\kappa \in (\kappa_0-\delta, \kappa_0+\delta)$ with small $\delta>0$ and all $y \in \{0,1\}$ are bounded and Lipschitz. By boundedness, their product is also Lipschitz. Conclude that for all $n \geq n_0^{\prime \prime}$ there exists a $c_0>0$ depending only on $\delta$ such that
\begin{align*}
\left|u_n(X_1, y, \widetilde{\beta}_{K_n}, \kappa) - v_n(X_1, y, \widetilde{\beta}_{K_n}, \kappa) \right| \leq c_0 \left|p_1(\widetilde{\beta}_{K_n}) - p_1(\beta_0) \right| \leq c_0 \left\| \widetilde{\beta}_{K_n} - \beta_0 \right\|,
\end{align*}
with probability one. Moreover, (A2) and  the Schwarz inequality imply that
\begin{align*}
\mathbb{E} \left\{\left|p_1(\widetilde{\beta}_{K_n})-p_1(\beta)  \right| \right\} \leq c_0 \mathbb{E}\left\{\left|\langle  X_1, \widetilde{\beta}_{K_n} - \beta \rangle \right| \right\} \leq c_0 \pi(\widetilde{\beta}_{K_n}, \beta).
\end{align*}
These two bounds jointly imply \eqref{eq:A13}, thus completing the first part of our proof.

To complete the proof we now establish \eqref{eq:A10} and for this we argue as in \citet{Kal:2023}. First note that on $F_{\delta}$ we have $\|\beta-\widetilde{\beta}_{K_n}\| \leq \delta + o(1) \leq 2 \delta$ for all large $n$, by (A4) and the fact that the uniform norm dominates the $\mathcal{L}^2([0,1])$-norm. Thus,
\begin{align}
\label{eq:A17}
\left|\frac{U_n(\widetilde{\beta}_{K_n}, \widehat{\beta}_n, \widehat{\kappa}_n)}{\gamma_n \pi\left(\widetilde{\beta}_{K_n}, \widehat{\beta}_n\right) \vee\gamma_n^2 } \right| \leq \sup_{\substack{\beta \in \Theta_{K_n}: \|\beta-\widetilde{\beta}_{K_n}\| \leq 2 \delta \\ |\kappa-\kappa_0|\leq \delta} } \left|\frac{U_n(\widetilde{\beta}_{K_n}, \beta, \kappa)}{\gamma_n \pi\left(\widetilde{\beta}_{K_n}, \beta\right) \vee\gamma_n^2 } \right|,
\end{align}
for all large $n$ with high probability. We will show that the random variable on the RHS of \eqref{eq:A17} is bounded in probability, which is equivalent to showing that
\begin{align}
\label{eq:A18}
\lim_{T \to \infty}\limsup_{n \to \infty} \mathbb{P}\left( \sup_{ \substack{ \beta \in \Theta_{K_n, \delta}, \kappa \in V_{\delta} \\ \pi\left(\widetilde{\beta}_{K_n}, \beta \right) \leq \gamma_n }} \left|U_n(\widetilde{\beta}_{K_n}, \beta, \alpha) \right| \geq T \gamma_n^2 \right) = 0,
\end{align}
as well as
\begin{align}
\label{eq:A19}
\lim_{T \to \infty}\limsup_{n \to \infty} \mathbb{P}\left( \sup_{ \substack{ \beta \in \Theta_{K_n, \delta}, \kappa \in V_{\delta} \\ \pi\left(\widetilde{\beta}_{K_n}, \beta \right) > \gamma_n }} \left|\frac{U_n(\widetilde{\beta}_{K_n}, \beta, \kappa)}{\pi\left(\widetilde{\beta}_{K_n}, \beta \right)} \right| \geq T \gamma_n \right) = 0,
\end{align}
where for notational simplicity we have defined the balls
\begin{align*}
\Theta_{K_n, \delta} & = \left\{\beta \in \Theta_{K_n}: \| \beta - \widetilde{\beta}_{K_n}\| \leq 2 \delta \right\} \\
V_{\delta} &= \left\{ \kappa \in \mathbbm{R}: |\kappa-\kappa_0|\leq \delta \right\}.
\end{align*}

For the proof of both \eqref{eq:A18} and \eqref{eq:A19} we will use the fact that for all $\epsilon>0$ sufficiently small, say $\epsilon \leq \epsilon_0$, there exists a $B>0$ such that
\begin{align*}
\int_{0}^{\epsilon} \log^{1/2}\left(1+ \frac{1}{u} \right) du \leq B \epsilon \log^{1/2}\left(\frac{1}{\epsilon} \right).
\end{align*}
(Use, e.g., L'Hospital's rule). To show \eqref{eq:A18} we will apply Theorem 5.11 of \citet{van de Geer:2000} on the process $U_n(\widetilde{\beta}_{K_n}, \beta, \kappa)$ for $\beta \in \Theta_{K_n, \delta}$, $\kappa \in V_{\delta}$ and $\pi(\beta, \widetilde{\beta}_{K_n}) \leq \gamma_n$. To prepare for this application, write $U_n(\widetilde{\beta}_{K_n}, \beta, \kappa)$ as
\begin{align*}
U_n(\widetilde{\beta}_{K_n}, \beta, \kappa) = \int z_{\beta, \kappa}  d\left(\mathbb{P}_n-\mathbb{P}_{X,Y} \right) = n^{-1/2} v_n(z_{\beta, \kappa}),
\end{align*}
where $z_{\beta, \kappa}$ and $\mathbb{P}_{X,Y}$ are as in Lemma~\ref{lem:5} and $v_n(\cdot)$ denotes the empirical process.  By our definitions, $z_{\beta, \kappa} \in \mathcal{Z}(n, 2 \delta, \delta)$ and each $z_{\beta, \kappa}$ is uniformly bounded by $2+1/(\kappa_0-\delta)$. Let $\mathcal{H}_{B, K}(\epsilon, \mathcal{Z}(n, 2 \delta, \delta), \mathbb{P}_{X,Y})$ denote the generalized entropy  with bracketing (in the Bernstein "norm" $\rho_K$) of the class of functions $\mathcal{Z}(n, 2 \delta, \delta)$, as given in Definition 5.1 in \citet{van de Geer:2000}. By the uniform boundedness of the elements of  $\mathcal{Z}(n, 2 \delta, \delta)$, Lemma 5.10 of \citet{van de Geer:2000} yields
\begin{align}
\label{eq:A20}
\mathcal{H}_{B, 8+4/(\alpha_0-\delta)}\left(\epsilon, \mathcal{Z}(n, 2 \delta, \delta) \right) \leq H_{B}(\epsilon/\sqrt{2}, \mathcal{Z}(n, 2 \delta, \delta), \mathbb{P}_{X,Y}), \quad \epsilon>0.
\end{align}
Observe next that by our assumptions for all $\|\beta-\widetilde{\beta}_{K_n}\| \leq 2 \delta$ we have $|H(\langle X, \beta \rangle)) - H(\langle X, \beta_0 \rangle))| \leq c_0 \delta + o(1)$, as $n \to \infty$. As discussed, previously, $H(\langle X, \beta_0 \rangle)$ is bounded away from $0$ and $1$ with probability one. Hence, $H(\langle X, \beta \rangle))$ is also bounded away from $0$ and $1$ for all $\beta \in \Theta_{K_n, \delta}$  when $\delta$ is chosen small enough and $n$ is large with probability one. From this we may deduce that for all $\kappa \in (\kappa_0-\delta, \kappa_0+\delta)$ and $y \in \{0,1\}$, $f_{H(\langle X, \beta \rangle)}^{\kappa}(y)$ and $f_{H(\langle X, g \rangle)}^{1+\kappa}(y)$ are Lipschitz with respect to the probabilities $H(\langle X, \beta \rangle)$ for $\|\beta-\widetilde{\beta}_{K_n}\| \leq 2 \delta$. A straightforward calculation now shows that there exists a global $c_0>0$ such that
\begin{align*}
\left|z_{\beta, \kappa}(X,Y)\right| \leq c_0 \left|\langle X, \beta- \widetilde{\beta}_{K_n} \rangle \right|.
\end{align*}
Taking expectations yields
\begin{align}
\label{eq:A21}
\mathbb{E}\left\{\left|z_{\beta, \kappa}(X,Y)\right|^2 \right\} \leq c_0^2 \left|\pi\left(\beta, \widetilde{\beta}_{K_n}\right)\right|^2 \leq c_0^2 \gamma_n^2.
\end{align}
It follows from Lemma 5.8 of \citet{van de Geer:2000} that we may take $R_n = c_0 \gamma_n$ in Theorem 5.11 of \citet{van de Geer:2000}. Assumption (A5) implies that $R_n \to 0$ as $n \to \infty$. Hence $R_n \leq \epsilon_0$ for all large $n$. Thus, taking square roots, integrating \eqref{eq:A20} from $0$ to $R_n$ and applying Lemma~\ref{lem:5}, we get
\begin{align*}
\int_{0}^{R_n} \mathcal{H}_{B, 8+4/(\alpha_0-\delta)}^{1/2}(u, \mathcal{Z}(n, 2 \delta, \delta) du \leq c_0 K_n^{1/2} \gamma_n \log^{1/2} (n) = c_0 \frac{K_n \log(n)}{n^{1/2}}.
\end{align*}
Conditions (5.31), (5.33) and (5.34) of Theorem 5.11 \citet{van de Geer:2000} are satisfied if we take $C_0 = T/c_0$ and $C_1 = (8+4/(\alpha_0-\delta))T/c_0^2$. An application of that theorem then yields
\begin{align*}
\mathbb{P}\left( \sup_{ \substack{ \beta \in \Theta_{K_n, \delta}, \kappa \in V_{\delta} \\ \pi\left(\widetilde{\beta}_{K_n}, \beta \right) \leq \gamma_n }} \left|U_n(\widetilde{\beta}_{K_n}, \beta, \kappa) \right| \geq T \gamma_n^2 \right) & = \mathbb{P}\left( \sup_{ \substack{ \beta \in \Theta_{K_n, \delta}, \kappa \in V_{\delta} \\ \pi\left(\widetilde{\beta}_{K_n}, \beta \right) \leq \gamma_n }} \left|v_n(z_{\beta,\kappa}) \right| \geq T n^{1/2} \gamma_n^2 \right)
\\ & \leq C \exp \left[ - \frac{T^2 K_n \log(n)}{C^2(C_1+1)}  \right].
\end{align*}
for some global $C>0$. The exponential tends to $0$, as, by (A4), $K_n \log n \to \infty$. Furthermore, as this holds for all sufficiently large $T$, \eqref{eq:A18} holds and therefore the first step of our proof is complete.

Finally, we establish \eqref{eq:A19}. We begin by observing that, by (A5), $\pi(\beta, \widetilde{\beta}_{K_n}) \leq C\|\beta - \widetilde{\beta}_{K_n} \| \leq 2 C \delta \leq \epsilon_0/c_0$ whenever $\beta \in \Theta_{K_n, \delta}$ with $\delta \leq \epsilon_0/(c_0 2C)$. Here, $c_0$ is the constant appearing in \eqref{eq:A21}. Therefore, it suffices to prove
\begin{align}
\label{eq:A22}
\lim_{T \to \infty}\limsup_{n \to \infty} \mathbb{P}\left( \sup_{ \substack{ \beta \in \Theta_{K_n, \delta}, \kappa \in V_{\delta} \\ \gamma_n <\pi\left(\widetilde{\beta}_{K_n}, \beta \right) \leq \epsilon_0/c_0 }} \left|\frac{U_n(\widetilde{\beta}_{K_n}, \beta, \kappa)}{\pi\left(\widetilde{\beta}_{K_n}, \beta \right)} \right| \geq T \gamma_n \right) = 0.
\end{align}
Let $S = \min\{s>1: 2^{-s} \epsilon_0/c_0 < \gamma_n \}$. By (A4), $K_n \asymp n^{\gamma}$ for some $\gamma \in (0,1)$. Hence, we can have $S \leq \lceil c_0 \log_2(n) + 1 \rceil$ for some $c_0>0$. Boole's inequality gives
\begin{align*}
\mathbb{P}\left( \sup_{ \substack{ \beta \in \Theta_{K_n, \delta}, \kappa \in V_{\delta} \\ \gamma_n <\pi\left(\widetilde{\beta}_{K_n}, g \right) \leq \epsilon_0/c_0 }} \left|\frac{U_n(\widetilde{\beta}_{K_n}, \beta, \kappa)}{\pi\left(\widetilde{\beta}_{K_n}, \beta \right)} \right| \geq T \gamma_n \right) & \leq  \sum_{s=1}^S  \mathbb{P}\left( \sup_{ \substack{ \beta \in \Theta_{K_n, \delta}, \kappa \in V_{\delta} \\ 2^{-s}\epsilon_0/c_0 <\pi\left(\widetilde{\beta}_{K_n}, \beta \right) \leq 2^{-s+1}\epsilon_0/c_0 }} \left|\frac{U_n(\widetilde{\beta}_{K_n}, \beta, \kappa)}{\pi\left(\widetilde{\beta}_{K_n}, \beta \right)} \right| \geq T \gamma_n \right)
\\ & \leq  \sum_{s=1}^S  \mathbb{P}\left( \sup_{ \substack{ \beta \in \Theta_{K_n, \delta}, \kappa \in V_{\delta} \\ \pi\left(\widetilde{\beta}_{K_n}, \beta \right) \leq 2^{-s+1}\epsilon_0/c_0 }} \left|U_n(\widetilde{\beta}_{K_n}, \beta, \alpha) \right| \geq T 2^{-s}\gamma_n \epsilon_0/c_0 \right)
\end{align*}
Clearly, for every $s \geq 1$ we have $2^{-s+1} \epsilon_0 \leq \epsilon_0$. Hence, by \eqref{eq:A21} and Lemma~\ref{lem:5},
\begin{align*}
\int_{0}^{2^{-s+1} \epsilon_0}\mathcal{H}^{1/2}_{B, 8+4/(\alpha_0-\delta)}\left(u, \mathcal{Z}(n, 2 \delta, \delta) \right) du \leq 2^{-s+1} B^{\star} K_n^{1/2} \log^{1/2}(n), 
\end{align*}
for some global $B^{\star}>0$, where we have  used the fact that $2^{s} \leq 2^{S} \leq n^{2c}$ for all large $n$. We will apply Theorem 5.11 of \citet{van de Geer:2000} individually to each summand, see also the proof of \eqref{eq:A18}. The conditions of that theorem are satisfied for $C_0 = T\epsilon_0/(2B^{\star}c_0)$ and $C_1 = T(8+4/(\alpha_0-\delta))/c_0 $ with $T$ sufficiently large, as, by definition of $S$, we have $\gamma_n \leq 2^{-s+1} \epsilon_0/c_0$ for all $ s = 1, \ldots, S$.  Therefore, by Theorem 5.11 of \citet{van de Geer:2000}, there exists a global constant $C{^\prime}>0$ such that
\begin{align*}
\mathbb{P}\left( \sup_{ \substack{ \beta \in \Theta_{K_n, \delta}, \kappa \in V_{\delta} \\ \pi\left(\widetilde{\beta}_{K_n}, \beta \right) \leq 2^{-s+1}\epsilon_0/c_0 }} \left|U_n(\widetilde{\beta}_{K_n}, \beta, \kappa) \right| \geq T 2^{-s}\gamma_n \epsilon_0/c_0 \right) \leq C^{\prime} \exp \left[-\frac{T^2 K_n \log(n)}{4|C^{\prime}|^2(C_1+1)c_0^2} \right].
\end{align*}
The exponential does not depend on $s =1, \ldots, S$, hence upon summing and recalling that $S \leq \lceil c_0 \log_2(n) + 1 \rceil$ for some $c_0>0$ we arrive at
\begin{align*}
\mathbb{P}\left( \sup_{ \substack{ \beta \in \Theta_{K_n, \delta}, \kappa \in V_{\delta} \\ \gamma_n <\pi\left(\widetilde{\beta}_{K_n}, \beta \right) \leq \epsilon_0/c_0 }} \left|\frac{U_n(\widetilde{\beta}_{K_n}, \beta, \kappa)}{\pi\left(\widetilde{\beta}_{K_n}, \beta \right)} \right| \geq T \gamma_n \right) \leq C^{\prime} \lceil c_0  \log_2(n) + 1 \rceil  \exp \left[-\frac{T^2 K_n \log(n)}{4|C^{\prime}|^2(C_1+1)c_0^2} \right],
\end{align*}
for all large $n$. For all $T$ large enough, say $T \geq T_0$, the RHS tends to $0$, as $n \to \infty$, by virtue of $K_n \asymp n^{\gamma}$. We have thus established \eqref{eq:A19}, thereby completing the proof.
\end{proof}

\section*{References}

\begingroup
\renewcommand{\section}[2]{}%

\endgroup
\end{document}